\documentclass[11pt,noparskip]{article}
\usepackage{style}
\usepackage{authblk}

\usepackage{physics}
\numberwithin{equation}{section}

\newcommand{\U}{\mathrm{U}}
\newcommand{\SU}{\mathrm{SU}}
\newcommand{\SL}{\mathrm{SL}}
\newcommand{\EL}{\mathrm{EL}}
\newcommand{\Mat}{\mathrm{Mat}}
\newcommand{\CNOT}{\mathrm{CNOT}}
\newcommand{\Sym}{\mathrm{Sym}}
\newcommand{\CAT}{\mathrm{CAT}}

\title{Depth-$1$ expanders on the unitary group and applications 
}
\date{}
\author[1]{Anurag Anshu}
\author[2]{Shankar Balasubramanian}
\author[3]{Jonas Haferkamp}
\author[4]{Aram W. Harrow}
\author[5]{Xinyu Tan}
\makeatletter
\@namedef{@sep4}{,\protect\\[0.25em]}
\makeatother
\affil[1]{School of Engineering and Applied Sciences, Harvard University, Cambridge, MA 02138, USA}
\affil[2]{Walter Burke Institute for Theoretical Physics and Institute for Quantum Information and Matter,
  California Institute of Technology, Pasadena, CA 91125, USA}
  \affil[3]{Faculty of Computer Science, Ruhr-University Bochum, Bochum, 44801, Germany}
  \affil[4]{Center for Theoretical Physics -- a Leinweber Institute, Massachusetts Institute of Technology, Cambridge, MA 02139, USA}
  \affil[5]{Department of Mathematics, Massachusetts Institute of Technology, Cambridge, MA 02139, USA}

\begin{document}

\maketitle
\vspace{-2em}
\begin{abstract}
We construct a constant-degree and constant-gap quantum expander on $n$
qubits where each unitary can be implemented by a depth-$1$ and 1D circuit of
Pauli or CNOT gates.  We provide two applications of this expander.
First, we use it to construct a family of frustration-free 1D
Hamiltonians whose ground states obey the entanglement-gap relation $S
= \Theta(\Delta^{-1/2})$; this is believed to be optimal, but
achieving it had been open.  Second, we use it to provide a streaming
protocol that tests for closeness to a class of 1D volume-law
entangled states.

Moreover, we extend our quantum expander to a constant-degree and constant-gap expander on the unitary group where each unitary is a single $T$ gate, a single $T^{\dagger}$ gate, or a depth-$1$ Clifford circuit. 
This implies that a random sequence of unitaries from the expander yields a gapped walk on a dense subgroup of the unitary group.  This improves upon previous work by Bourgain and Gamburd which did not control the dependence of the gap on the dimension.
\end{abstract}

\setcounter{tocdepth}{2}
\tableofcontents

\section{Introduction}\label{sec:intro}

Quantum expanders are collections of unitaries defined via a spectral gap property: 
\begin{definition}[Quantum expanders]\label{def:q_expander_channel}
  A finite set of unitaries $\{U_1,\ldots,U_m\}\subset \U(d)$ is called a \emph{quantum expander} with {\em gap} $\Delta>0$ if
  \begin{equation}\label{eq:gap-cond}
  \norm{\frac 1m \sum_i U_i X U_i^\dag }_2 \leq (1-\Delta)\norm{X}_2 \qfor X \qq{s.t.} \Tr X = 0.
  \end{equation}
  The {\em degree} is the number of unitaries $m$.
  An \emph{expander family} is a sequence of quantum expanders all with the same $m$ and $\Delta$ but with an increasing series of dimensions $d_1,d_2,\ldots$; i.e.\ with constant gap and degree.
\end{definition}

They were first introduced in~\cite{ben2007quantum, ben2010quantum, hastings2007entropy} and found initial applications in constructing short quantum one-time pads~\cite{ambainis2004small}, showing that the quantum entropy difference problem is QSZK-complete~\cite{ben2007quantum}, and giving examples of quantum states with decay of correlations and volume-law entanglement~\cite{hastings2007entropy}.
Afterwards, \cite{aharonov2014local} used quantum expanders to construct optimal communication protocols for testing highly entangled states and counterexamples to the ``graph area law'' conjecture.

While quantum expanders can be constructed if $U_i$ is a random unitary~\cite{PhysRevA.76.032315}, there is also interest in efficient constructions where $U_i$ is a circuit on $n$ qubits of depth $\poly(n)$.  Such constructions exist using a variety of combinatorial tricks~\cite{ben2007quantum, gross2007quantum, harrow2007quantum}, but it has been an open question to provide general lower bounds on the depth.  In this paper, we show that there surprisingly exists a quantum expander family whose unitaries can be implemented by a \emph{depth-$1$} circuit, even on a one-dimensional geometry:
\begin{theorem}[Informal; see \Cref{thm:expander_gap}]
There exists a family of quantum expanders with a constant spectral gap, and the defining unitaries $U_1, U_2, \ldots, U_{16}$ are depth-$1$ circuits of Pauli or $\CNOT$ gates.
\end{theorem}
This result follows the work of Kassabov~\cite{Kas07a}, who provides a generating set for the group of CNOT gates obeying an expander property.  
We observe that each generator can be written as a depth-$1$ circuit of CNOT gates on a 1D local geometry.  
Upon adding a few more constant-range generators, we show that the CNOT expander can be promoted to a quantum expander.

If we allow for general Clifford gates (although still depth-$1$ in 1D) and a single T gate we can prove a much stronger result by leveraging the recently obtained spectral gap for random Pauli rotations~\cite{baer2026random,HLT25} and Kassabov's bound~\cite{Kas07a}.
We first define an \textit{expander on the unitary group}:
\begin{definition}[Expander on the unitary group]
For any representation $\rho$ of $\SU(d)$, define $\Phi_\rho=\frac{1}{m}\sum_{i=1}^{m}\rho(V_i)$.
We call a collection of unitaries $V_1,\ldots,V_m\in \SU(d)$ an expander on $\rm \SU(d)$ with spectral gap $\Delta>0$ if 
\begin{equation}
  \norm*{\Phi_{\rho }|X\rangle}\leq (1-\Delta)\norm*{|X\rangle}
  \label{eq:rho-gap}
\end{equation}
for all $|X\rangle$ orthogonal to the invariant subspace of $\rho$ and all finite dimensional unitary representations $\rho$.
We call $m$ the degree of the expander.
As in \cref{def:q_expander_channel}, an expander family is an infinite sequence of expanders on dimensions $d_1 < d_2 < \cdots$ with the same degree that obey \cref{eq:rho-gap} for the same $\Delta>0$.
\end{definition}
We then prove the following result in~\Cref{appendix:unitary}:
\begin{theorem}[see also \Cref{thm:unitary-expander}]\label{thm:short-TPE}
There exists an explicit expander $\mathsf{V}=\{V_1,\ldots,V_{25}\}$ on $\SU(2^n)$ of degree $25$ with a constant spectral gap.
Moreover, each unitary is a single $T$ or $T^\dagger$ gate or a depth-$1$ Clifford circuit.
\end{theorem}

Taking the representation $\rho$ to be $\rho(U) = U \ot \overline{U}$,
\cref{eq:rho-gap} reduces to the quantum expander condition in
\cref{def:q_expander_channel}.  It has long been known~\cite{Arnold62}
that for any universal gate set $\{V_i\}$, \cref{eq:rho-gap} holds for
any fixed $\rho$ with $\Delta>0$ possibly depending on $\rho$.
However, establishing \cref{eq:rho-gap} uniformly in the choice of
$\rho$ is not so simple.  The proof that most Haar-random unitaries
give quantum expanders~\cite{PhysRevA.76.032315} was extended in
\cite{HH08} to show that they also give tensor-product expanders.
This means that \cref{eq:rho-gap} holds for $\Delta$ depending only on
degree $m$ for the representation $\rho_t = U^{\ot t} \ot
\overline{U}^{\ot t}$, uniformly for $t \lesssim 2^{n/6}$.  On the
other hand, in \cite{bourgain2012spectral,bourgain2008spectral},
\cref{eq:rho-gap} was proven to hold uniformly in $\rho$ (or equivalently for $t\rightarrow\infty$) for all
universal gate sets in $\SU(d)$ with algebraic entries.  However,
their work does not show the existence of particular gate sets for
which we can control the scaling of $\Delta$ with $d$.  Our
\cref{thm:short-TPE} is the first result that establishes a gap that
is uniform in both the irrep $\rho$ and dimension $d = 2^n$.  At the
expense of a larger degree ($\leq 5400$) we extend the family of
generators to $\SU(d)$ for all $d\geq 2$
in~\cref{cor:unitary-expander}.  The existence of such a uniformly
gapped family of generators was conjectured by Lubotzky as
communicated by Bourgain in~\cite{bourgain2017random}.  The closest
analogous result is Kassabov's expander for the Cayley graph of the
symmetric group~\cite{kassabov2007symmetric} and generalizations to
finite simple groups~\cite{kassabov2006finite}, which provide similar
uniform gap bounds.

An immediate corollary of \cref{thm:short-TPE} is obtained by sampling sequences from $\{V_1,\ldots,V_{25}\}$.  This yields an $\varepsilon$-approximate $k$-design in depth $O(nk+\log(1/\varepsilon))$ (see e.g.~\cite{brandao2016local} or~\cite{schuster2025strong} for strong approximate designs). 
Like the constructions in~\cite{o2023explicit}, these approximate designs match the cardinality lower bound $2^{\Omega(nk)}$.

In the rest of the paper, we provide some observations and applications of the depth-$1$ quantum expander.
\paragraph{Gap vs convergence rate of Matrix Product Operators.}
The matrix $M = \frac{1}{16}\sum_{i=1}^{16} U_i \ot \overline{U}_i
$ corresponds to the mixed unitary channel by applying
a random $U_i$.  Because of the 1D structure of the unitaries, $M$ is
also a Matrix Product Operator (MPO)~\cite{CPSV17} with $O(1)$ bond
dimension.  Matrix Product Operators arise frequently in tensor
network calculations.  For example, contracting an infinite
translationally invariant 2D tensor network is closely related to the
task of finding the top eigenvalue of the MPO
corresponding to one column or row of the original tensor
network~\cite{Lubasch_2014,Fishman_2018}.  A key algorithmic question
is how quickly power iteration will converge to the corresponding
eigenvector.  If the second singular value is $1-\Delta$ then this will
require $O(n/\Delta)$ steps~\cite{chi2-divergence} and if an entropy
decay rate constant $\alpha$ (related to the MLSI constant) is known
then this implies~\cite{KastoryanoT13} a bound of $O(\log(n)/\alpha)$.
It is natural to suspect that the $O(n)$ prefactor in the
$O(n/\Delta)$ bound is often unnecessary. For example, two-point
correlations decay with length scale $O(1/\Delta)$ without the need
for the prefactor.  See \cite{Hastings14} for discussion of this
prefactor in related contexts.

However, our matrix $M$ is an
example where the gap $\Delta$ is constant (due to the expander
condition), and yet it requires $O(n)$ steps to converge to its top
eigenvalue.  This can be seen easily in the channel picture.  The
stationary state, corresponding to an eigenvector with eigenvalue 1,
is the maximally mixed state on $n$ qubits, which has rank $2^n$.  But each application of
the channel can increase the state's rank by at most 16.  If we start with a
pure state we will still be far from the stationary distribution after
$<n/4$ iterations.  This slow convergence also means that the entropy decay constant must be significantly different from the gap; in this case we must have $\alpha \leq \log(n)/n$.

\vspace{0.5cm}

\noindent In \cref{sec:vollaw} and \cref{sec:streaming}, we discuss two further applications to the 1D area law and streaming entanglement testing.

\paragraph{1D area law.}
The area law for quantum many-body systems with a spectral gap states that the bipartite entanglement entropy of the ground state $\rho_{\Omega} = \ketbra{\Omega}{\Omega}$, denoted $S_A(\rho_{\Omega})$, satisfies $S_A(\rho_{\Omega}) \propto |\partial A|$.  This was first proved in 1D by Hastings \cite{hastings2007entropy}, with the generalization to higher dimensions an open question.  Hastings' proof provided the bound $S \lesssim \exp(1/\Delta)$ where $\Delta$ is the spectral gap (assuming a constant local dimension on each site).  This bound was improved in ~\cite{arad2013area} to $S \lesssim 1/\Delta$, which was conjectured to be optimal since the correlation length scales as $\lesssim 1/\Delta$ \cite{Hastings04, hastingskoma} : 
\begin{conjecture}\label{conj:area}
There exists a family of frustrated 1D Hamiltonians $\{H_n\}_{n}$ with gap $\Theta(n^{-1})$ and a unique ground state with entanglement entropy $\Theta(n)$.
\end{conjecture}
For frustration-free Hamiltonians, the bound on entanglement entropy is believed to be $S \lesssim 1/\Delta^{1/2}$.  While there is no proof of this, there is strong evidence coming from the relationship between correlation length and gap $\xi \lesssim 1/\Delta^{1/2}$ in frustration-free systems proved in \cite{gosset2016correlation} and from local gap thresholds proved in \cite{gosset2016local, anshu2020improved, lemm2024critical}.  

As a main application of our expander construction, we show that the conjectured entanglement-gap scaling for frustration-free systems can be achieved: 
\begin{theorem}[Informal; see \Cref{thm:hamfamily}]\label{thm:ent-gs}
There exists a family of frustration-free 1D Hamiltonians $\{H_n\}_n$ with gap $\Theta(n^{-2})$ and a unique ground state with entanglement entropy $\Theta(n)$.
\end{theorem}

Our construction of the Hamiltonian family is inspired by the construction in \cite{aharonov2014local}.  In particular, our Hamiltonian is split into left and right 1D spin chains.  Each spin chain is divided into two subchains corresponding to clock and data chains.  On the left spin chain, we apply $U_i$ on the data chain and on the right spin chain we apply $U_i^{\dagger}$; both of these use a modification of the Feynman-Kitaev domain wall clock.  Since $U_i$ are depth-$1$, they can be applied while maintaining 1D locality.  Adding a specifically designed Hamiltonian term coupling the left and right chains causes the ground state on the data chain to be a maximally entangled state with entanglement $\sim n$.  By padding the circuit for $U_i$ with identities, we can provably achieve gap $\sim 1/n^2$.  We emphasize that we must use a \emph{modified} domain wall clock due to both the padding step and the choice of the coupling Hamiltonian.

We briefly comment on previous work.  Call $\alpha$ the exponent relating entanglement and gap for a given infinite family of Hamiltonians, i.e.\ $S \sim \Delta^{-\alpha}$.  All constructions in the literature have $\alpha$ significantly smaller than $1/2$.  Many constructions have $S = \Theta(n)$ but $\Delta$ exponentially or superexponentially small in $n$~\cite{ippoliti2025infinite,klich}.  Gottesman and Hastings~\cite{gottesman2010entanglement} provided a construction where $\alpha = 1/4$ and Irani~\cite{irani2010ground} provided a construction where $\alpha = 1/12$; both were mildly frustrated Hamiltonians.  Other constructions achieve $\alpha \approx 1/5$ with $S$ exhibiting sub-volume law entanglement~\cite{caha2018pair, movassagh2016supercritical}.  Thus, our construction is the first to achieve $\alpha = 1/2$ and saturate the conjectured bound for frustration-free Hamiltonians.

\paragraph{Streaming algorithms for state testing.}  Another application of the depth-$1$ quantum expander concerns state testing of highly entangled states in the streaming setting.  This provides a new resource for characterizing quantum states, as well as a nontrivial class of protocols.

An analogous classical setting to this problem is the communication complexity of the equality function.  Suppose Alice and Bob have bit strings $x_A$ and $x_B$ of length $n$ and want to test $x_A \stackrel{?}{=} x_B$.  For deterministic protocols with worst-case guarantees, $n+1$ bits of communication are needed.  If Alice and Bob are (1) allowed to utilize randomness and (2) allowed a failure probability $1/n$, only $O(\log(n))$ bits of communication are needed.  This protocol can then be promoted to a streaming setting, where the length-$2n$ input $x_A x_B$ is streamed and the verifier checks $x_A \stackrel{?}{=} x_B$ with success probability $1-1/n$ only storing $O(\log n)$ bits in memory.

We promote this classical result to the quantum setting: here, we would like a test for the entangled state $\ket{\Gamma} = 2^{-n/2} \sum_{x \in \{0,1\}^n} \ket{x} \otimes \ket{x}$.  However, a naive generalization of the classical test will fail because it cannot detect relative phases in the superposition.  Instead, we utilize the expander based test in~\cite{aharonov2014local}, which improves upon prior results~\cite{bennett1996mixed, barnum2002authentication, harrow2008communication} and provides a protocol using $O(1)$ qubits or cbits of communication.  For the communication protocol, the expander just needed to be efficient; however, to convert this to a streaming test, we need the expander to be low depth and have an almost translation-invariant structure.  Our depth-$1$ quantum expander has this property, leading to:
\begin{theorem}[Informal; see \Cref{sec:streaming}]
There is a single pass streaming test with $O(\Delta^{-1} \log(1/\epsilon))$ memory that implements the POVM $\{M, I-M\}$ with $\norm{M - \ketbra{\Gamma}{\Gamma}} \leq \epsilon$.
\end{theorem}

In fact, $\ket{\Gamma}$ can be replaced by the state $U \ket{\Gamma}$ where $U$ is a translation-invariant matrix product operator with constant bond dimension without affecting the memory cost of the protocol.  This provides a class of highly-entangled states which can be efficiently tested in the streaming setting.

\paragraph{Improved gap, entanglement, and correlation bounds.}
An important principle in quantum many-body physics is that ground states of gapped local Hamiltonians inherit the locality of their parent Hamiltonians.  Developing quantitative area law bounds (which control the scaling of $S$ with $\Delta$) or sharp bounds on ground state correlation decay gives us mathematically precise statements that embody this principle.

Our \cref{thm:ent-gs} shows the limits of any possible 1D area law by
illustrating a frustration-free example where $S =
\Theta(1/\Delta^{1/2})$.  Can we prove an improved area law bound which this example saturates?  This question remains open, but in
\cref{append:localgapscaling}, we prove that $S \leq
\tilde{O}(1/\Delta_{\mathrm{loc}}^{3/4})$ for a frustration-free
Hamiltonian, where $\Delta_{\mathrm{loc}}$ is the local gap of the 1D
system.  The proof uses the Approximate Ground Space Projector (AGSP)
constructed in \cite{AAG22} using the Sherstov's robust
polynomial~\cite{She12}.
Previously, the best known bound was $S \leq \tilde{O}(1/\Delta) = \tilde{O}(1/\Delta_{\mathrm{loc}})$ \cite{arad2013area}.

In \cref{append:CATlb}, we provide bounds on the correlation decay of ground states and their gaps, which hold for general $k$-local Hamiltonians and provide stronger bounds than in Ref.~\cite{hastingskoma}.  One application is that we show if the ground state $\psi_{AB}$ obeys $\norm{\psi_A - \ketbra{\mathrm{CAT}}{\mathrm{CAT}}_A} \leq 1/10$, where $\ket{\mathrm{CAT}}$ is a cat state, then the gap of the Hamiltonian must be $\leq O(1/n)$.  A similar bound can be shown for the kind of clock states studied in \cref{sec:vollaw}.  Since these kinds of states are required to apply the quantum expander unitaries coherently, such bounds provide quantitative limits on whether constructions like in \cref{sec:vollaw} can generate highly entangled ground states with large gap.

\paragraph{Outlook.}  Our work leaves several open questions.  One natural question is Conjecture~\ref{conj:area}.  Our constructions rely on a clock Hamiltonian, which behaves like a single-particle walk with gap $1/n^2$.  If the single particle Hamiltonian can be replaced with a \emph{many-particle} Hamiltonian, such as a free-fermion Hamiltonian, the gap could be amplified to $1/n$ while preserving the volume-law entanglement.  Unfortunately, since each particle carries a local clock, more terms need to be added to synchronize the local clocks, which will reduce the gap.  An alternative construction was provided in \cite{ippoliti2025infinite} which maps ground states of Hamiltonians to $+1$ eigenstates of staircase/sequential unitaries.  Using free-fermion Hamiltonians, this construction could achieve a gap $1/n$, but all of our attempts failed to generalize the construction from unitaries to Hamiltonians.

Another open question concerns further generalization of the streaming protocol for state testing.  It would be interesting to test for more natural states like ground states of translation-invariant Hamiltonians.  The general principle would be to construct an operator of the form $H = \sum_{j=1}^m U_j$ where $U_j$ is a (nearly) translation-invariant staircase circuit so that $H$ has a $+1$ eigenstate $\ket{\psi}$ and a gap to all other eigenstates.  This would provide a streaming algorithm with $O(\log m)$ memory.

Finally, we wonder if there are any other interesting applications of the depth-$1$ quantum expander or the depth-$1$ expander on $\rm \SU(2^n)$.  
For example, a key ingredient in the breakthrough result MIP*=RE~\cite{ji2021mip} is a device-independent self-test for many EPR pairs that utilizes Pauli-braiding and quantum low-degree tests.  However, unlike our quantum expander test, this self-test does not assume that Pauli measurements are implemented correctly. Because of this, the protocol requires a (polynomially) large question and answer alphabet, and afterwards compression techniques are used to reduce its size. If the CNOT-based observables appearing in our expander could be tested efficiently, our construction could yield an efficient entanglement test with a natively small question and answer alphabet.

\paragraph{Notations.}
We use $\norm{A}$ for the operator norm of a matrix $A$ and $\overline{A}$ for the complex conjugate of $A$. 
$\CNOT_{a,b}$ refers to the CNOT gate controlled by qubit $a$ and acting on qubit $b$.  Whenever we write a Hamiltonian term $\ketbra{a}{a}_i$, we mean $I_1\otimes I_2\otimes\cdots \otimes \ketbra{a}{a}_i \otimes I_{i+1}\otimes\cdots $

\paragraph{Concurrent work}
While preparing this manuscript we became aware of concurrent work~\cite{liu2026almost}, which observes nearly identical circuits for Kassabov's generators.
Moreover, they also leverage this result to obtain a constant-depth uniformly gapped generator set for $\mathrm{Cl}(n)$ using Nikolov's product decomposition~\cite{nikolov2005product} of which the Aaronson-Gottesmann~\cite{aaronson2004improved} decomposition is a special case. 
Their construction requires more generators and is not of depth-$1$ but of constant depth, but makes the generalization to all $n$ (as opposed to $n=3s$ ) explicit. 
Also they construct circuits on an open 1D line, while ours are on a circle.
Last,~\cite{liu2026almost} also extends the result to higher $k$ but exploits the gap of the PFC ensemble~\cite{metger2024simple} instead of the gap of Pauli rotations~\cite{baer2026random}.
As a result, unlike us, they do not obtain a uniformly gapped generator set for all representations. 

\section{Depth-\(1\) quantum expanders}\label{sec:expander}

In this section, we use a definition of quantum expanders equivalent to \Cref{def:q_expander_channel}: 
\begin{definition}[Quantum expanders]
    We say that a set of $n$-qubit unitary matrices $\{U_1, \ldots, U_m\}$ forms a \emph{quantum expander} on $n$ qubits with a spectral gap of $\Delta > 0$ if the second largest singular value of $\frac{1}{m}\sum_i U_i \otimes \overline{U_i}$ equals $1-\Delta$. 
\end{definition}

Below we state the main theorem of this section. 

\begin{theorem}[Depth-$1$ constant-sized quantum expanders with constant spectral gaps]\label{thm:expander_gap}
    For any integer $s\geq 1$, there exists a set of $16$ unitary matrices on $3s$ qubits, denoted as $\{U_1, \ldots, U_{16}\}$, such that 
    it forms a quantum expander on $3s$ qubits with a spectral gap of at least $1/(2 \cdot 10^8)$,  
    and each $U_i$ can be written as a depth-$1$ circuit. 
    Moreover specifically, 
    \begin{itemize}
        \item each $U_1, \ldots, U_6$ acts non-trivially only on the first $3$ qubits: 
        \begin{align}
            U_1 = X\otimes I_{2^{3s-1}}, &\quad U_2 = Z \otimes I_{2^{3s-1}}, \\
            U_3 = \CNOT_{1,2} \otimes I_{2^{3s-2}}, & \quad U_4 = \CNOT_{2,1}\otimes I_{2^{3s-2}} ,\\
            U_5 = I_2 \otimes \CNOT_{2,3} \otimes I_{2^{3s-3}}, & \quad U_6 = I_2 \otimes \CNOT_{3,2} \otimes I_{2^{3s-3}}. 
        \end{align}
        
        \item each $U_7\ldots, U_{16}$ is translation-invariant by $3$ consecutive qubits: 
        \begin{align}
            U_7, \ldots, U_{12} &= (\CNOT_{a,b})^{\otimes s} = \prod_{i=0}^{s-1} \CNOT_{3i+a,3i+b}, \quad \text{ for }1\leq a\neq b\leq 3. \quad \\
            U_{13} &= \prod_{i=0}^{s-1} \CNOT_{3i+4,3i+2}, \quad U_{14} = \prod_{i=0}^{s-1} \CNOT_{3i+5, 3i+1} \\
            U_{15} &= \prod_{i=0}^{s-1} \CNOT_{3i+5, 3i+3}, \quad U_{16} = \prod_{i=0}^{s-1} \CNOT_{3i+6, 3i+2},
        \end{align}
        where the indices are taken with periodic condition. 
    \end{itemize}
\end{theorem}
As an example, see \Cref{fig:U13-cnot-layer} for the circuit diagram of $U_{13}$. The rest of the section is to prove \Cref{thm:expander_gap}.

\begin{figure}[ht]
    \centering
\begin{tikzpicture}[x=1cm,y=0.7cm, every node/.style={font=\small}]
  \tikzset{
    qubit/.style={fill=black,circle,inner sep=1pt},
    cnot target/.style={
      draw,circle,inner sep=0pt,minimum size=11pt,
      path picture={
        \draw[line width=0.5pt]
          (-0.27,0) -- (0.27,0)
          (0,-0.27) -- (0,0.27);
      }
    }
  }

  \draw[rounded corners] (-0.8,1.4) rectangle (9.8,-1.4);

  \coordinate (q1)    at (0,  1);
  \coordinate (q2)    at (0,  0);
  \coordinate (q3)    at (0, -1);

  \coordinate (q4)    at (2,  1);
  \coordinate (q5)    at (2,  0);
  \coordinate (q6)    at (2, -1);

  \coordinate (q7)    at (4,  1);
  \coordinate (q8)    at (4,  0);
  \coordinate (q9)    at (4, -1);

  \coordinate (q3nm5) at (7,  1);
  \coordinate (q3nm4) at (7,  0);
  \coordinate (q3nm3) at (7, -1);

  \coordinate (q3nm2) at (9,  1);
  \coordinate (q3nm1) at (9,  0);
  \coordinate (q3n)   at (9, -1);

  \coordinate (q1ext) at (11,  1);
  \coordinate (q2ext) at (11,  0);
  \coordinate (q3ext) at (11, -1);

  \foreach \q in {q1,q2,q3,q4,q5,q6,q7,q8,q9,q3nm5,q3nm4,q3nm1,q3nm3,q3nm2,q3n,q1ext,q2ext,q3ext}{
    \node[qubit] at (\q) {};
  }

  \node[left=4pt] at (q1) {$1$};
  \node[left=4pt] at (q2) {$2$};
  \node[left=4pt] at (q3) {$3$};

  \node[left=4pt] at (q4) {$4$};
  \node[left=4pt] at (q5) {$5$};
  \node[left=4pt] at (q6) {$6$};

  \node[left=4pt] at (q7) {$7$};
  \node[left=4pt] at (q8) {$8$};
  \node[left=4pt] at (q9) {$9$};

  \node[left=4pt] at (q3nm5) {$3s-5$};
  \node[left=4pt] at (q3nm4) {$3s-4$};
  \node[left=4pt] at (q3nm3) {$3s-3$};

  \node[left=4pt] at (q3nm2) {$3s-2$};
  \node[left=4pt] at (q3nm1) {$3s-1$};
  \node[left=4pt] at (q3n)   {$3s$};

  \node[right=4pt] at (q1ext) {$1$};
  \node[right=4pt] at (q2ext) {$2$};
  \node[right=4pt] at (q3ext) {$3$};

  \node at (5,  1) {$\cdots$};
  \node at (5,  0) {$\cdots$};
  \node at (5, -1) {$\cdots$};

  \draw (q4)    -- node[pos=1, cnot target, sloped] {} (q2);
  \draw (q7)    -- node[pos=1, cnot target, sloped] {} (q5);
  \draw (q3nm2) -- node[pos=1, cnot target, sloped] {} (q3nm4);
  \draw (q1ext) -- node[pos=1, cnot target, sloped] {} (q3nm1);

\end{tikzpicture}
\caption{The circuit of $U_{13}$ with $s$ CNOT gates on $3s$ qubits with periodic boundary. }
\label{fig:U13-cnot-layer}
\end{figure}
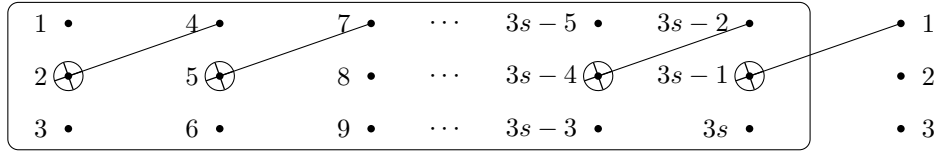

\subsection{Proof idea}

One group that will play a central role in the section is the special linear group over $\F_2$, consisting of all invertible $n\times n$ binary matrices, for which we will denote by $\SL(n; \F_2)$. 
This group is closely related to the group generated by all CNOT gates on $n$ qubits, for which we will denote by $G_n$. 
Let us first clarify how they are related. 

Each $g\in \SL(n;\F_2)$ acts linearly on $\F_2^n$, and hence defines a permutation of the $2^n$ bitstrings. This yields a natural representation
\begin{equation}
    \Gamma: \SL(n;\F_2) \to \U(2^{n})
\end{equation}
by letting $g$ permute computational basis states. 
Concretely, for any $g\in \SL(n;\F_2)$ and $x\in \F_2^{n}$, 
\begin{equation}
    \Gamma(g)\cdot \ket{x} = \ket{g\cdot x}. 
\end{equation}

Let $e_i\in \F_2^n$ be the standard basis vector with a $1$ in the $i$-th coordinate and zeros elsewhere. 
$\SL(n; \F_2)$ admits a generating set of all elementary matrices of form
\begin{equation}
    E_{i,j} = I_n + e_i e_j^T, \qquad 1\leq i\neq j\leq n. 
\end{equation}
For $x\in \F_2^n$, we have $E_{i,j}x = x + x_j e_i$, i.e.\ the action of $E_{i,j}$ is to add the $j$-th bit $x_j\in \F_2$ into the $i$-th coordinate (mod $2$). Under $\Gamma$, this is exactly the action of a CNOT gate with control qubit $j$ and target qubit $i$: 
\begin{equation}\label{eq:elementary_matrix_CNOT}
    \Gamma(E_{i,j}) = \CNOT_{j,i}. 
\end{equation}
Hence $\Gamma(\SL(n;\F_2)) = G_n$, the group generated by all CNOT gates on $n$ qubits. 

Our proof of \Cref{thm:expander_gap} follows from two ingredients: 
First, we will show that, if we have a constant-sized ``$2$-design on $G_n$'' with a constant spectral gap, then by adding the single-qubit Pauli $X$ and $Z$ gates to the set, it forms a quantum expander on $\U(2^n)$ with a constant gap. 
This step follows from how Harrow and Hastings proved that a constant-sized $2$-design on the symmetric group of order $2^n$ can be augmented to a quantum expander with a constant gap.\footnote{Harrow and Hastings added a more complicated gate to the permutation $2$-design; in fact, their proof also works if one simply takes the additional gate as a single-qubit Pauli $Z$. } 

The second ingredient is to find a constant-sized ``$2$-design on $G_n$'' with a constant spectral gap, for which we mean that a set of matrices in $G_n$, $\{P_1, \ldots, P_m\}$, satisfying
\begin{equation}\label{eq:Gn_2_design_gap}
    \norm{\frac{1}{m} \sum_{i=1}^m \rho(P_i) - \E_{\bP \sim \mu(G_{n})} \rho(\bP)} \leq 1 - \Delta,
\end{equation}
where $\rho$ is the representation $\rho: G_{n} \ni P \mapsto P^{\otimes 2}$, $\mu(\cdot)$ denotes the uniform distribution, and $m$ and $\Delta > 0 $ are two universal constants. 
Equivalently, we would like to find a generating set $\{A_1, \ldots, A_m\}$ of $\SL(n;\F_2)$ such that
\begin{equation}\label{eq:SL_design_gap}
    \norm{\frac{1}{m} \sum_{i=1}^m \rho(\Gamma(A_i)) - \E_{\bA \sim \mu(\SL(n;\F_2))} \rho(\Gamma(\bA))} \leq 1 - \Delta,
\end{equation}
where $\rho\cdot \Gamma$ is a representation of $\SL(n;\F_2)$. 

It was shown by Kassabov in \cite[Theorem 8a]{Kas07a} that for each $n=3s$ for some integer $s\geq 1$, there exists a generating set of $14$ elements for $\SL(3s; \F_2)$ with a bounded Kazhdan constant, which can be translated to a constant spectral gap via \cite[Lemma 2.17]{CHHLMT24}.\footnote{Combining these two results in fact shows a stronger result that \cref{eq:Gn_2_design_gap} holds for every finite-dimensional unitary representation $\rho$ of $G_n$. So to some extent, the Kassabov's generators may seem to be an overkill. Yet it is unclear to us whether there exist simpler constructions if one only concerns about a particular representation of $\SL(n;\F_2)$, for example the $\rho\Gamma$ we use. } 
Interestingly, due to \cref{eq:elementary_matrix_CNOT},
each of the $14$ generators given by Kassabov when converted to a circuit on $n$ qubits under the representation $\Gamma$ admits a depth-$1$ circuit of CNOT gates. Moreover, $4$ of them are a single CNOT gate and the rest of them are translation-invariant by $3$ consecutive qubits. 

Combining these two ingredients yields \Cref{thm:expander_gap}.

\subsection{Depth-\(1\) expanders on the group of CNOTs}

Recall that $G_n$ is the group generated by all CNOT gates on $n$ qubits. The main theorem of this subsection is the following. 
\begin{theorem}\label{thm:gap_from_kas}
    For any integer $s\geq 1$, let $U_3, \ldots, U_{16}$ be the unitary matrices on $3s$ qubits defined in \Cref{thm:expander_gap}. 
    Then for any finite-dimensional unitary representation $\rho$ of $G_{3s}$, we have that
    \begin{equation}
        \norm{\frac{1}{14} \sum_{i=3}^{16} \rho(U_i) - \E_{\bP \sim \mu(G_{3s})} \rho(\bP)} \leq 1 - \frac{1}{48\cdot 10^5}. 
    \end{equation}
\end{theorem}

The proof of \Cref{thm:gap_from_kas} has three ingredients. 
First, the Kazhdan constant bounds the largest eigenvalue of the difference of the averaging operators \cite[Lemma 2.17]{CHHLMT24} (restated as \Cref{lem:kazhdan_gaps_relation}). 
Second, Kassabov supplies a constant-sized generating set of $\SL(3s;\F_2)$ whose Kazhdan constant is uniformly bounded below \cite[Theorem 8a]{Kas07a} (restated as \Cref{thm:Kas_SL}). 
Finally, we bound its smallest eigenvalue away from $-1$ using a short odd relation among three of the generators (\Cref{lem:lower_bound_min_eval}).

Consider a group $G$ generated by a set $S$. We write $\calK(G;S)$ for the \emph{Kazhdan constant for $G$ with respect to $S$}. 
In this paper, we will not be concerned with the specific definition of the Kazhdan constant; interested readers can find more discussions in \cite[Section 2.4]{CHHLMT24} and \cite{Kas07a}. 
What matters to us is how the Kazhdan constant relates to the spectral gap:
\begin{lemma}[Bounding the largest eigenvalue with the Kazhdan constant {\cite[Lemma 2.17]{CHHLMT24}}]\label{lem:kazhdan_gaps_relation}
    Let $G$ be a compact group (finite or Lie) generated by a subset $S$ which is closed under taking inverses. Then for any finite-dimensional unitary representation $\rho$ of $G$, 
    \begin{equation}\label{eq:kazhdan_gap_max_eval}
        \lambda_{\max}\left(\frac{1}{\abs{S}}\sum_{g\in S} \rho(g) - \E_{\bg\sim \mu(G)} \rho(\bg)\right) \leq 1 - \frac{\calK(G;S)^2}{2\abs{S}}, 
    \end{equation}
    where $\lambda_{\max}(\cdot)$ denotes the largest eigenvalue. Furthermore, if $S$ contains the identity, then $\cref{eq:kazhdan_gap_max_eval}$ holds with $\lambda_{\max}(\cdot)$ replaced with the operator norm $\norm{\cdot}$. 
\end{lemma}

\begin{remark}
The full statement of \cite[Lemma 2.17]{CHHLMT24} gives an operator-norm bound by controlling both the largest and smallest
eigenvalues. 
Its assumption that the generating set contains the identity is used to bound the smallest eigenvalue away from $-1$.
We could therefore adjoin the identity to our generating family, at the cost of one additional unitary in the final quantum
expander. 
However, for the particular Kassabov family considered here, this is unnecessary because we can directly bound the smallest eigenvalue as shown in \Cref{lem:lower_bound_min_eval}. 
\end{remark}

We will also use the following two lemmas in~\Cref{appendix:unitary}.
\begin{lemma}[\cite{kassabov2007symmetric}]\label{lemma:kazhdan-by-subgroups}
    Let $H_i$ be subgroups of $G$ and $S_i\subseteq H_i$. 
    Then, we have
    \begin{equation}
        \mathcal{K}\left(G;\bigcup_i S_i\right)\geq \frac{1}{2} \mathcal{K}\left(G;\bigcup_i H_i\right) \inf_i \mathcal{K}\left(H_i;S_i\right).
    \end{equation}
\end{lemma}
For a proof see~\cite[Lemma 2.14]{CHHLMT24}
\begin{lemma}[Short product lemma~\cite{kassabov2007symmetric}]\label{lemma:short-product}
Let $S$ and $S'$ be two generating sets of a group $G$ such that every element of $S'$ can be written as a product of at most $k$ elements of $S$. Then,
\begin{equation}
    \mathcal{K}(G;S)\geq \frac{1}{k}\mathcal{K}(G;S').
\end{equation}
\end{lemma}

Note that $G_{3s} = \Gamma(\SL(3s;\F_2))$.
We now show that for each integer $s\geq 1$, there exists a generating set $\Sigma_{3s}$ of $\SL(3s;\F_2)$ such that $\calK(\SL(3s;\F_2); \Sigma_{3s})$ is lower bounded by a constant $>0$ and $\Sigma_{3s}$ has constantly many elements. 
This is precisely the statement of \cite[Theorem 8a]{Kas07a} (restated as \Cref{thm:Kas_SL}). 

We first set up some notations and facts that eventually build up to the definition of the generating set $\Sigma_{3s}$. 
Let $R$ be a unital ring that may be noncommutative. 
An \emph{elementary matrix group} over $R$ of rank $m$, denoted by $\EL(m;R)$, is the multiplicative matrix group generated by all elementary matrices 
\begin{equation}
    E_{a,b}(r) = I_m + r e_{a,b}, \quad \text{ where }1\leq a\neq b \leq m \text{ and }r\in R. 
\end{equation}
Let $\Mat(s;\F_2)$ be the matrix algebra consisting of all $s\times s$ matrices over $\F_2$. 
We have the following facts, whose proofs can be found in \cite[Section 5.1]{CHHLMT24}:
\begin{fact}
    $\SL(3s;\F_2) = \EL(3s;\F_2) = \EL(3;\Mat(s; \F_2))$. 
\end{fact}
\begin{fact}
    $\Mat(s;\F_2)$ can be generated by $2$ elements: 
    \begin{align}
        A = \begin{pmatrix}
            0 & 1 &   &   &   &    \\
              & 0 & 1 &   &   &    \\
              &   &   & \ddots & \ddots & \\
              &   &   &        &   0    & 1\\
            1 &   &   &   &   &   0          
        \end{pmatrix}, \quad
        B = \begin{pmatrix}
            1 & 0 &   &   &   &    \\
            0 & 0 &  &   &   &    \\
              &   & 0 &  &  & \\
              &   &   &  \ddots &       & \\
             &   &   &   & 0  &  
        \end{pmatrix} . \label{eq:matricesAB}
    \end{align}
\end{fact}

\begin{fact}\label{fact:SL_generators}
    $\EL(3;\Mat(s;\F_2))$ can be generated by the following $14$ elements, which form the set $\Sigma_{3s}$: 
    \begin{align}
        E_{a,b}(I_s) \text{ where }1\leq a \neq b\leq 3, \quad\text{and}\quad E_{a,b}(A), \ E_{a,b}(B) \text{ where } |a-b|=1.
    \end{align}
\end{fact}

\begin{theorem}[{\cite[Theorem 8a]{Kas07a}}]\label{thm:Kas_SL}
    For any integer $s\geq 1$, let $\Sigma_{3s}$ be the set of $14$ matrices in \Cref{fact:SL_generators}. Then
    \begin{equation}
        \calK(\SL(3s;\F_2); \Sigma_{3s}) > 1/400. 
    \end{equation}
\end{theorem}
We remark that the original Kassabov's statement says $28$ generators. But because the field is $\F_2$, the generators are in fact involutions. So effectively it is just $14$ generators. 

It remains to show that the smallest eigenvalue is bounded away from $-1$ by a constant. 
\begin{lemma}[Bounding the smallest eigenvalue]\label{lem:lower_bound_min_eval}
    Let $G = \SL(3s;\F_2)$ and $S = \Sigma_{3s}$. 
    For any finite-dimensional unitary representation $\rho$ of $G$, 
    \begin{align}
        \lambda_{\min} \left(\frac{1}{\abs{S}} \sum_{g\in S} \rho(g)\right) \geq -1 + \frac{1}{63} ,
    \end{align}
    where $\lambda_{\min}(\cdot)$ denotes the smallest eigenvalue. 
\end{lemma}
\begin{proof}
    Let $x=E_{1,2}(I_s)$, $y=E_{2,3}(I_s)$, and $z=E_{1,3}(I_s)$. Since $x,y,z$ are involutions and $[x,y] = z$, we have that $xyxyz = I_{3s}$ over $\F_2$. Set $A = \frac{1}{\abs{S}} \sum_{g\in S} \rho(g)$. Since $S$ is closed under taking the inverses, $A$ is Hermitian. 
    Suppose that $Av = \lambda v$ for some unit vector $v$ and $\lambda\in \R$. 
    Our goal is to show that $\lambda \geq -1+\frac{1}{63}$. 
    
    For any $g\in S$, write $a_g = \norm{(\rho(g)+I)v}$. Telescoping along the relation $xyxyz=e$ gives
    
    \begin{align}
        2 &= \norm{(\rho(xyxyz) + I)v} \notag \\
        &= \norm{(\rho(xyxyz) + \rho(xyxy) - \rho(xyxy) -  \rho(xyx) + \rho(xyx) + \rho(xy) - \rho(xy) -\rho(x) + \rho(x) + I)v} \notag\\
        &\leq 2a_x + 2a_y + a_z \tag{triangle inequality} \\
        &\leq \sqrt{4 + 4 + 1}\sqrt{a_x^2 + a_y^2 + a_z^2} \tag{Cauchy--Schwarz} \\
        &\leq 3\sqrt{\sum_{g\in S} a_g^2} \tag{$x,y,z\in S$}. 
    \end{align}
    On the other hand,
    \begin{equation}
        \sum_{g\in S} a_g^2 = \sum_{g\in S} \norm{(\rho(g) + I)v}^2 = \sum_{g\in S} v^\dagger (2 I + \rho(g) + \rho(g)^\dagger) v = 2\abs{S} \cdot (1+\lambda)
    \end{equation}
    Since $\abs{S} = 14$, then $\lambda \geq (2/3)^2/28 - 1 = -1 + 1/63$. 
\end{proof}

We now have all three ingredients to prove \Cref{thm:gap_from_kas}. 
\begin{proof}[Proof of \Cref{thm:gap_from_kas}]
    For each integer $s\geq 1$, we first show that the unitaries $U_3, \ldots, U_{16}$ as defined in \Cref{thm:expander_gap} are exactly the operator in $\Gamma(\Sigma_{3s})$. 

    Recall that each element in $\SL(3s;\F_2)$ under the representation $\Gamma$ gives rise to a product of CNOT gates on $3s$ qubits. 
    Suppose we index the qubits by $1,2,\ldots, 3s$. 
    Our goal to find the circuits for each of the generators in $\Sigma_{3s}$. 
    Let us start with the easiest one: 
    $E_{a,b}(B)$ is the identity matrix plus a sole $1$ at position $((a-1)s+1, (b-1)s+1)$, the top left entry in the $(a,b)$ block. So
    \begin{equation}
        \Gamma(E_{a,b}(B)) = \CNOT_{(b-1)s+1,(a-1)s+1}. 
    \end{equation}
    Similarly, each $E_{a,b}(I_s)$ and $E_{a,b}(A)$ 
    is a product of $s$ CNOT gates with non-overlapping targets and controls, i.e.\ a depth-$1$ circuit. More specifically,
    \begin{align}
        \Gamma(E_{a,b}(I_s)) &= \CNOT_{(b-1)s + 1, (a-1)s+1}\CNOT_{(b-1)s + 2,(a-1)s+2}\cdots \CNOT_{bs,as}, \\
        \Gamma(E_{a,b}(A)) &= \CNOT_{(b-1)s+2, (a-1)s+1}\CNOT_{(b-1)s + 3, (a-1)s+2}\cdots \CNOT_{bs, as-1}\CNOT_{(b-1)s + 1, as}.  
    \end{align}
    Now, if we index the $3s$ qubits by $1, 4, \ldots, 3s-2, 2, 5, \ldots, 3s-1, 3, 6, \ldots, 3s$ instead of $1,2,\ldots, 3s$, then these circuits become exactly $U_3, \ldots, U_{16}$. 

    For any finite-dimensional unitary representation $\rho$ of $G_{3s}$,
    we have that 
    \begin{equation}
        M \coloneqq \frac{1}{14}\sum_{i=3}^{16} \rho(U_i) - \E_{\bP\sim \mu(G_{3s})} \rho(\bP) = \frac{1}{14}\sum_{g\in \Sigma_{3s}} \rho(\Gamma(g)) - \E_{\bA\sim \mu(\SL(3s;\F_2))} \rho(\Gamma(\bA)).
    \end{equation}
    By \Cref{lem:kazhdan_gaps_relation,thm:Kas_SL}, we have that
    \begin{equation}
        \lambda_{\max}(M) \leq 1 - \frac{\calK(\SL(3s;\F_2);\Sigma_{3s})^2}{2\cdot 14} \leq 1 - \frac{1}{400^2\cdot 28} . 
    \end{equation}
    It follows from \Cref{lem:lower_bound_min_eval} that $\lambda_{\min}(M)\geq -1 +\frac{1}{63}$, which completes the proof. 
\end{proof}

\subsection{Augment generating sets of the group of CNOTs to quantum expanders}

Denote by $\Sym(N)$ the symmetric group of order $N$. Then $G_n$, the group generated by all CNOT gates on $n$ qubits, is a subgroup of $\Sym(2^n)$.
Below is the main theorem of this subsection.
\begin{theorem}\label{thm:expander_CNOTs_to_unitary}
    Let $n\geq 2$. 
    Suppose that $P_1, \ldots, P_m\in G_n$ satisfy that
    \begin{equation}\label{eq:CNOTs_moment_two}
        \norm{\frac{1}{m} \sum_{i=1}^m P_i \otimes P_i - \E_{\bP \sim \mu(G_n)}  \bP \otimes \bP} \leq 1-\epsilon. 
    \end{equation}
    Let $P_{m+1} = X\otimes I_{2^{n-1}}$ and $P_{m+2} = Z \otimes I_{2^{n-1}}$. Then
    $\{P_1, \ldots, P_{m+2}\}$ forms a quantum expander on $n$ qubits with a spectral gap of at least $\frac{\min(m\epsilon, 2)}{24(m+2)}$. 
\end{theorem}

\begin{remark}
The condition in \cref{eq:CNOTs_moment_two} concerns only the
representation \(P\mapsto P^{\otimes 2}\), whereas
\Cref{thm:gap_from_kas} provides a uniform spectral gap for every
finite-dimensional representation of \(G_n\). Thus,
\Cref{thm:gap_from_kas} is stronger than necessary for
\Cref{thm:expander_CNOTs_to_unitary}, and we do not know whether the
required bound for \(P\mapsto P^{\otimes 2}\) admits a simpler proof.
Its full strength is nevertheless useful in
\Cref{appendix:unitary}, where we construct quantum tensor product
expanders at all moments.

The augmentation argument above does not directly extend to moments
\(t>1\). Already for \(t=2\), the relevant representation of the CNOT
group is \(P\mapsto P^{\otimes 4}\). The affine permutation group
generated by Pauli \(X\) and CNOT gates is an exact permutation
\(3\)-design~\cite[Proposition 2.15]{HRT25} but not a permutation \(4\)-design. Consequently, adjoining only the
single-qubit Pauli \(X\) and \(Z\) gates does not suffice to repeat the
preceding argument at higher moments. The construction in
\Cref{appendix:unitary} instead first passes to the Clifford group and
then adjoins a T gate.
\end{remark}

\begin{lemma}
    The set of $2^n \times 2^n$ matrices that commute with every CNOT gate is given by
    \begin{equation}
        \Span \{I_{2^n}, \ketbra{+^n}{+^n}, \ketbra{0^n}{0^n}, \ketbra{0^n}{+^n}, \ketbra{+^n}{0^n}\}. 
    \end{equation}
\end{lemma}
\begin{proof}
    Recall that the representation $\Gamma$ of $\SL(n;\F_2)$ given by
    \begin{equation}
        \Gamma(g) \ket{x} = \ket{gx}, \quad \text{ for every }g\in \SL(n;\F_2), \ x\in \F_2^n,
    \end{equation}
    satisfies that $\Gamma(\SL(n;\F_2))=G_n$. Therefore, a matrix $M$ commutes with every CNOT gate if and only if
    \begin{equation}
        M\cdot \Gamma(g) = \Gamma(g)\cdot M, \quad \text{ for every }g\in \SL(n;\F_2). 
    \end{equation}
    Write $M = \sum_{x,y\in \F_2^n} M_{x,y} \ketbra{x}{y}$. It is further equivalent to 
    \begin{align}
        \sum_{x,y\in \F_2^n} M_{x,y} \ketbra{x}{y} &= \sum_{x,y\in \F_2^n} M_{x,y} \Gamma(g)\ketbra{x}{y}\Gamma(g)^\dagger \\
        &= \sum_{x,y\in \F_2^n} M_{x,y} \ketbra{gx}{gy} \\
        &= \sum_{x,y\in \F_2^n} M_{gx,gy} \ketbra{x}{y}, \quad \text{ for every }g\in \SL(n;\F_2). 
    \end{align}     
    This means that each matrix element $M_{x,y}$ can only depend on which orbit the pair $(x,y)$ belongs to under the action 
    \begin{equation}
        (x, y)\mapsto (gx,gy). 
    \end{equation}
    Thus the commutant of $G_n$ is the set of matrices whose entries are constant on each orbit of $\SL(n;\F_2)$ acting on ordered pairs $(x,y)$. 
    We now show that the orbits of $\SL(n;\F_2)$ acting on $\F_2^n \times \F_2^n$ are
    \begin{itemize}
        \item $(0,0)$
        \item $(0,x)$ where $x\neq 0$
        \item $(x,0)$ where $x\neq 0$
        \item $(x,x)$ where $x\neq 0$
        \item $(x,y)$ where $x,y\neq 0$ and $x\neq y$
    \end{itemize}
    These $5$ sets are disjoint and their union is $\F_2^n \times \F_2^n$. 
    Each of the first $4$ sets is an orbit since $\SL(n;\F_2)$ acts transitively on $\F_2^n\setminus \{0\}$. 
    For any $(x,y)$ where $x,y\neq 0$ and $x\neq y$, because $x$ and $y$ are linearly independent over $\F_2$, 
    there always exists $g\in \SL(n;\F_2)$ such that $gx = e_1$ and $gy = e_2$ (think in terms of Gaussian elimination). This shows that the last set is an orbit.

    Therefore, the commutant of $G_n$ is the linear span of 
    \begin{align}
        A_1 = \ketbra{0^n}, \quad A_2 = \sum_{x\in \{0,1\}^n, x\neq 0^n} \ketbra{0^n}{x}, \quad A_3 = \sum_{x\in \{0,1\}^n, x\neq 0} \ketbra{x}{0^n}, \\
        A_4 = \sum_{x\in \{0,1\}^n, x\neq 0^n} \ketbra{x}{x}, \quad \text{and}\quad A_5 = \sum_{x,y\in \{0,1\}^n, x,y\neq 0, x\neq y} \ketbra{x}{y},
    \end{align}
    which is the same as 
    \begin{equation}
        \Span \{I_{2^n}, \ketbra{+^n}{+^n}, \ketbra{0^n}{0^n}, \ketbra{0^n}{+^n}, \ketbra{+^n}{0^n}\}.
    \end{equation}
    This completes the proof. 
\end{proof}

\begin{lemma}[{\cite[Lemma 1]{HH08}}]
    \label{lem:kill_overlap}
    Let $\Pi$ be a projector and $A, B$ be Hermitian operators with $\norm{A}, \norm{B}\leq 1$. Suppose there exist $\epsilon_A, \epsilon_B\in (0,1)$ such that
    \begin{itemize}
        \item $\Pi A = A\Pi = \Pi$
        \item $\norm{(I - \Pi) A(I-\Pi)} \leq 1-\epsilon_A$
        \item $\norm{\Pi B\Pi}\leq 1-\epsilon_B$.
    \end{itemize}
    Then for any $p\in (0,1)$, we have that
    \begin{equation}
        \norm{pA + (1-p)B}\leq 1 - \frac{\epsilon_B}{12} \min (p\cdot \epsilon_A, 1-p). 
    \end{equation}
\end{lemma}

We are now ready to prove \Cref{thm:expander_CNOTs_to_unitary}. 

\begin{proof}[Proof of \Cref{thm:expander_CNOTs_to_unitary}]
    Let $N=2^n$. Note that $I_{N}$ is the only operator (up to scalar) that commutes with every element in $\U(N)$. 
    Let 
    \begin{equation}
        \ket{a_1} = \sum_{i\in [N]}\ket{i}\!\ket{i}. 
    \end{equation}
    Then
    \begin{equation}\label{eq:op_avg_unitary}
        \E_{\bU \sim \mu(\U(2^n))}  \bU \otimes \overline{\bU} = \frac{\ketbra{a_1}}{\norm{\ket{a_1}}^2}. 
    \end{equation}
    The set of $2^n \times 2^n$ matrices that commute with every CNOT gate is given by
    \begin{equation}
        \Span \{I_{2^n}, \ketbra{+^n}{+^n}, \ketbra{0^n}{0^n}, \ketbra{0^n}{+^n}, \ketbra{+^n}{0^n}\}. 
    \end{equation}
    Let $\ket{a_2} = \ket{+^n}\!\ket{+^n}$, $\ket{a_3} =  \ket{0^n}\!\ket{0^n}$, $\ket{a_4} = \ket{0^n}\!\ket{+^n}$ and $\ket{a_5} = \ket{+^n}\!\ket{0^n}$.
    Suppose that $\{\ket{b_i}\}_{i=1}^5$ is the orthonormal set of vectors obtained after applying the Gram-Schmidt process to $\{\ket{a_i}\}_{i=1}^5$. 
    Then $\ket{b_1} = \frac{\ket{a_1}}{\norm{\ket{a_1}}}$ and 
    \begin{equation}\label{eq:op_avg_CNOTs}
        \E_{\bP \sim \mu(G_n)}  \bP \otimes \bP = \sum_{i=1}^5 \ketbra{b_i}. 
    \end{equation}

    Let
    \begin{equation}
        \Pi = \E_{\bP \sim \mu(G_n)}  \bP \otimes \bP - \E_{\bU \sim \mu(\U(2^n))}  \bU \otimes \overline{\bU} = \sum_{i=2}^5 \ketbra{b_i}. 
    \end{equation}
    We will show that when $n\geq 2$, 
    \begin{equation}\label{eq:overlap_pi}
        \norm{\Pi\cdot (X_1 \otimes X_1 + Z_1 \otimes Z_1)\cdot \Pi} = \frac{N-2}{N-1}. 
    \end{equation}
    This will then complete the proof of \Cref{thm:expander_CNOTs_to_unitary} by applying \Cref{lem:kill_overlap} with $A = \frac{1}{m} \sum_{i=1}^m P_i \otimes P_i - \ketbra{b_1}$, $B = \frac12(X_1 \otimes X_1 + Z_1 \otimes Z_1) - \ketbra{b_1}$, $p = \frac{m}{m+2}$, $\epsilon_A = \epsilon$, and $\epsilon_B = 1 - \frac12 \cdot \frac{N-2}{N-1} \geq \frac12$. Then, 
    \begin{equation}
        \norm{\frac{1}{m+2} \sum_{i=1}^{m+2} P_i \otimes P_i - \E_{\bU \sim \mu(\U(2^n))}  \bU \otimes \overline{\bU}} = 
        \norm{pA + (1-p)B}\leq 1 - \frac{\min(m\epsilon, 2)}{24(m+2)} . 
    \end{equation}

    So it remains to prove \cref{eq:overlap_pi}. 
    Let $H = X_1 \otimes X_1 + Z_1 \otimes Z_1$. Then
    \begin{equation}
        \norm{\Pi\cdot (X_1 \otimes X_1 + Z_1 \otimes Z_1)\cdot \Pi} = \norm{\sum_{i,j=2}^5 \braket{b_i|H}{b_j}\cdot \ketbra{i}{j}}. 
    \end{equation}
    This is a $4\times 4$ matrix that can be calculated exactly but tediously, because $\ket{b_1},\ldots, \ket{b_5}$ are the vectors obtained after the Gram-Schmidt process. 
    Let us slightly simplify it. 
    Let $M = \sum_{i=1}^5 \ketbra{a_i}{i}$ and let $M=QR$ be its $QR$ decomposition. So $Q = \sum_{i=1}^5 \ketbra{b_i}{i}$, which satisfies $Q^\dagger Q = I_5$, and $R$ is upper triangular.
    Then
    \begin{equation}
        \sum_{i,j=1}^5 \braket{b_i|H}{b_j}\cdot \ketbra{i}{j} = Q^\dagger H Q. 
    \end{equation}
    Our goal is to calculate the eigenvalues of $Q^\dagger H Q$. Note that for any scalar $\lambda$ and vector $v$ such that
    \begin{equation}
        Q^\dagger H Q v = \lambda v, 
    \end{equation}
    we must have that
    \begin{equation}
        M^\dagger H M w = \lambda M^\dagger M w, \quad \text{where }w = R^{-1}v. 
    \end{equation}
    Both $M^\dagger H M$ and $M^\dagger M$ are $5\times 5$ matrices which can be calculated from $\ket{a_1},\ldots, \ket{a_5}$ directly. 
    Let $r = 1/\sqrt{N}$. Then
    \begin{equation}
        M^\dagger M = \begin{pmatrix}
            N & 1 & 1 & r & r\\
            1 & 1 & r^2 & r & r\\
            1 & r^2 & 1 & r & r\\
            r & r & r & 1 & r^2\\
            r & r & r & r^2 & 1
        \end{pmatrix}, \text{ and }
        M^\dagger H M = \begin{pmatrix}
            2N & 2 & 2 & 2r & 2r \\
            2 & 1 & 2r^2 & r & r\\
            2 & 2r^2 & 1 & r & r\\
            2r & r & r & 0 & 2r^2\\
            2r & r & r & 2r^2 & 0
        \end{pmatrix}. 
    \end{equation} 
    Then a simple symbolic calculation gives us
    \begin{equation}
        \det (M^\dagger H M - \lambda M^\dagger M) = -\frac{N-2}{N^4}\cdot(\lambda - 2)\cdot (\lambda(N-1) + 2)^2\cdot (\lambda(N-1) - N+2)^2. 
    \end{equation}
    So $Q^\dagger H Q$ has one eigenvalue of $2$, two eigenvalues of $-\frac{2}{N-1}$, and two eigenvalues of $\frac{N-2}{N-1}$.

    Note that $\ket{b_1}$ is an eigenvector of $H$ with eigenvalue $2$. 
    Since $\{\ket{b_i}\}_{i=1}^5$ is orthonormal, we have that $\braket{b_1|H}{b_j} = 0$ for $j=2,\ldots, 5$. 
    So $\ket{1}$ is the eigenvector of $Q^\dagger H Q$ with the eigenvalue $2$. 
    Since $\abs*{-\frac{2}{N-1}}\leq \frac{N-2}{N-1}$ when $n\geq 2$, this proves \cref{eq:overlap_pi}. 
\end{proof}

\subsection{Depth-$1$ expanders on the special unitary group}\label{appendix:unitary}

For a finite set or compact subgroup \(S\), let \(\mu(S)\) denote the uniform or Haar probability measure on $S$. We say the spectral gap of a family of expanders is uniformly gapped if there exists
\(\Delta>0\), independent of \(n\), such that for every finite-dimensional unitary representation \(\rho\)
of \(H_n\) without invariant vectors, 
\begin{equation}
    \left\|\mathbb E_{U\sim\mu(S_n)}\rho(U)\right\|
    \leq 1-\Delta . 
\end{equation}

Throughout this subsection, we use the convention that \(T:=e^{-\ii \pi Z/8}\in\SU(2)\).
We prove the following theorem.

\begin{theorem}\label{thm:unitary-expander}
Let \(n=3s\) with \(s\in\mathbb{N}\). There exists an explicit expander
\begin{equation}
    \mathsf{V}^{(n)}=\{V_1,\ldots,V_{25}\}\subset \SU(2^n),
\end{equation}
on the unitary group such that
$V_1=T\otimes I_{n-1}, V_2=T^{\dagger}\otimes I_{n-1}$
and every other \(V_i\) is a depth-$1$ 1D Clifford circuit with a uniformly bounded spectral gap. 
\end{theorem}
We will also prove the following Corollary, which resolves a conjecture by Lubotzky~\cite{bourgain2017random}:
\begin{corollary}\label{cor:unitary-expander}
    For all $d\geq 2$ there exists an explicit expander of degree $\leq 5400$ on $\SU(d)$ with a uniformly bounded spectral gap.
\end{corollary}

We leverage the recently obtained spectral gap for random Pauli
rotations~\cite{baer2026random}. We first construct a uniformly gapped
generating set for the Clifford group.

For \(n\geq 3\), let \(H_j\) and \(S_j\) denote the corresponding
single-qubit gates acting on qubit \(j\). We denote the finite \(n\)-qubit Clifford circuit group by
\begin{equation}
    \mathrm{Cl}(n)
    :=
    \left\langle
        H_j,S_j,\CNOT_{j,k}:
        j,k\in[n],\ j\neq k
    \right\rangle
    \leq \SU(2^n) . 
\end{equation}

\begin{lemma}\label{lemma:Clifford-expander}
Let \(n=3s\) with \(s\in\mathbb{N}\). There exists an explicit uniformly
gapped generating set
\begin{equation}
    \mathsf{W}^{(n)}=\{W_1,\ldots,W_{23}\}
    \subset \mathrm{Cl}(n),
\end{equation}
such that every \(W_i\) is a depth-$1$ Clifford circuit.
\end{lemma}

This immediately yields another constant-depth quantum expander. 
In
contrast to the expander in \Cref{thm:expander_gap}, this construction
requires both more generators and the use of Hadamard and \(S\) gates.

Conjugation on Pauli operators induces a homomorphism from
\(\mathrm{Cl}(n)\) onto the symplectic group \(\mathrm{Sp}(2n;\F_2)\).
Uniformly gapped generating sets for \(\mathrm{Sp}(2n;\F_2)\) are known to
exist~\cite{kassabov2006finite,lubotzky2011finite}. Here we instead show
how a gapped generating set for \(G_n = \Gamma(\SL(n;\F_2))\), the group of all CNOT gates, can be used to
construct one for the Clifford group using methods native to quantum
information theory. 
We believe that this makes the circuit structure of the generators more
transparent.

\begin{proof}[Proof of \Cref{lemma:Clifford-expander}]
Assume \(n=3s\) is even for simplicity. The odd case proceeds analogously. 
Aaronson and Gottesman~\cite{aaronson2004improved}
showed that every Clifford unitary admits, up to a global phase, a normal form
\begin{equation}
    \mathsf{H}\mathsf{C}\mathsf{P}\mathsf{C}\mathsf{P}\mathsf{C}
    \mathsf{H}\mathsf{P}\mathsf{C}\mathsf{P}\mathsf{C},
\end{equation}
where \(\mathsf{P}\) denotes a layer of single-qubit gates from $\langle S \rangle = \{I, S, Z, S^\dagger\}$,
\(\mathsf{H}\) denotes a layer of $I$ and Hadamard gates, and
\(\mathsf{C} \in G_n\) is an arbitrary CNOT circuit.
Define
\begin{align}
    K_{\mathrm{H},1}
        :=\langle
            H\otimes I_{n-1}
          \rangle, \qquad
    &K_{\mathrm{H},n/2}
        :=\langle
            H^{\otimes n/2}\otimes I_{n/2}
          \rangle,\\
    K_{\mathrm{S},1}
        :=\langle
            S\otimes I_{n-1}
          \rangle, \qquad
    &K_{\mathrm{S},n/2}
        :=\langle
            S^{\otimes n/2}\otimes I_{n/2}
          \rangle.
\end{align}

First note that every Hadamard layer can be written as a product of constantly many elements of $G_n, K_{\mathrm{H},1}, K_{\mathrm{H},n/2}$. 
Indeed, for \(A\subseteq[n]\), write $H_A:=\bigotimes_{j=1}^n H^{\mathbf{1}_{j\in A}}$. 
Since \(G_n\) contains all qubit permutations, the claim holds for $H_A$ for every $|A|=1$. 
It remains to show the claim when \(|A|\) is even. 
Observe that there exist \(B,C\subseteq[n]\) with $ |B|=|C|=n/2$ such that $A = B\triangle C$, where $\triangle$ denotes the symmetric difference.
Hence $ H_{A}=H_BH_C$ and $H_B,H_C$ is a product of two qubit permutations and $H^{\otimes n/2} \otimes I_{n/2}$. 

An analogous bounded-product decomposition works for the phase layers. 
Since $(HS)^3 = e^{\mathrm{i}\pi/4}I$, the residual global
phase costs only constantly many additional factors.
Combining these observations with Aaronson--Gottesman, there exists a constant
\(L\) such that
\begin{equation}
    \mathrm{Cl}(n)
    \subseteq
    \bigl(
        G_n
        \cup K_{\mathrm{H},1}
        \cup K_{\mathrm{H},n/2}
        \cup K_{\mathrm{S},1}
        \cup K_{\mathrm{S},n/2}
    \bigr)^{L}.
\end{equation}

Since $\mathcal{K}(G;G)\geq \sqrt{2}$ (see e.g.~\cite{kassabov2007symmetric}),
the short-product lemma, \Cref{lemma:short-product}, implies that there is an absolute constant $\kappa_0>0$ such that
\begin{equation}
    \mathcal{K}\!\left(
        \mathrm{Cl}(n);
        G_n
        \cup K_{\mathrm{H},1}
        \cup K_{\mathrm{H},n/2}
        \cup K_{\mathrm{S},1}
        \cup K_{\mathrm{S},n/2}
    \right)
    \geq \kappa_0. 
\end{equation}
Recall that \(G_n = \Gamma(\SL(3s;\F_2)\) and $\calK(\SL(3s;\F_2); \Sigma_{3s}) >1/400$ by \Cref{thm:Kas_SL}. Set
\begin{equation}
    \mathsf{W}^{(n)}
    :=
    \Gamma(\Sigma_{3s})
    \cup K_{\mathrm{H},1}
    \cup K_{\mathrm{H},n/2}
    \cup K_{\mathrm{S},1}
    \cup K_{\mathrm{S},n/2}.
\end{equation}
Since $\abs{\Sigma_{3s}} = 14$, then $|\mathsf{W}^{(n)}|=23$. Thus $\mathcal{K}\bigl(\mathrm{Cl}(n);\mathsf{W}^{(n)}\bigr) \geq \kappa_1$ for an absolute constant \(\kappa_1>0\). Every element of \(\mathsf{W}^{(n)}\) is a depth-$1$ Clifford circuit.

Finally, as $\mathsf W_n$ contains the identity, \Cref{lem:kazhdan_gaps_relation} converts the constant lower bound on the Kazhdan constant into a constant lower bound on the spectral gap.
\end{proof}

\begin{remark}
We expect the same argument to extend to all \(n\). Choose $m=3\left\lfloor\frac{n}{3}\right\rfloor$ and embed one copy of \(\mathrm{Cl}(m)\) into the first \(m\) qubits
and another into the last \(m\) qubits. The results
of~\cite{baer2026random} should imply a gap for the resulting product
walk and hence a constant Kazhdan constant. One could then apply
\Cref{lemma:short-product,lemma:kazhdan-by-subgroups} to the two
overlapping Clifford subgroups.
\end{remark}

We are now ready to prove \Cref{thm:unitary-expander}.
The proof will reveal that adding a single T gate (and its inverse) to $\mathsf{W}^{(n)}$ suffices to obtain a uniformly gapped generator set for $\SU(2^n)$. 

\begin{proof}[Proof of \Cref{thm:unitary-expander}]
Set \(d=2^n\). For a finite-dimensional unitary representation
    $\rho:\SU(d)\to U(\mathcal{H})$,
define its invariant subspace by
\begin{equation}
    \operatorname{Inv}(\rho)
    :=
    \left\{
       | \psi\rangle\in\mathcal{H}:
        \rho(U)|\psi\rangle=|\psi\rangle
        \text{ for every }U\in\SU(d)
    \right\}.
\end{equation}
For a probability measure \(\nu\) on \(\SU(d)\), define the
averaging operator
\begin{equation}
    M_\rho(\nu)
    :=
    \int_{\SU(d)}\rho(U)\,\mathrm{d}\nu(U)
\end{equation}
and the essential norm
\begin{equation}
    g(\nu):=\sup_{
        \rho: \,
        \operatorname{Inv}(\rho)=\{0\}
    }
    \bigl\|M_\rho(\nu)\bigr\|.
\end{equation}
For independent \(U_i\sim\mu_i\), let \(\mu_1*\mu_2\) denote the
distribution of \(U_1U_2\), the \textit{convolution} of $\mu_1$ and $\mu_2$. Then
\begin{equation}
    M_\rho(\mu_1*\mu_2)
    =
    M_\rho(\mu_1)M_\rho(\mu_2).
    \label{eq:moment-convolution}
\end{equation}

Let  $Z_1:=Z\otimes I_{n-1}$ and
denote by \(\nu_{\mathrm{PR}}\) the distribution of a random Pauli rotation $  C e^{\mathrm{i}\phi Z_1}C^\dagger$, where \(C\) is uniformly distributed over \(\mathrm{Cl}(n)\) and
\(\phi\) is uniformly distributed over \([-\pi,\pi)\). The spectral-gap
bound of~\cite{baer2026random} states that
\begin{equation}
    g(\nu_{\mathrm{PR}})
    \leq 15/16.
\end{equation}
Let $K_Z:=
\left\{
    e^{\mathrm{i}\phi Z_1}:
    \phi\in[-\pi,\pi)
\right\}$ and denote by \(\mu(S)\) the uniform probability measure when
\(S\) is finite and the normalized Haar measure when \(S\) is a compact subgroup.
From a change of variables we find $\mu(\mathrm{Cl}(n))*\mu(K_Z)*\mu(\mathrm{Cl}(n))=\nu_{\mathrm{PR}} *\mu(\mathrm{Cl}(n))$.
Consequently,
\begin{equation}
    M_\rho(\mu(\mathrm{Cl}(n))*\mu(K_Z)*\mu(\mathrm{Cl}(n)))
    =
    M_\rho(\nu_{\mathrm{PR}})
  M_{\rho}(\mu(\mathrm{Cl}(n))),
\end{equation}
and hence
\begin{equation*}
    g(\mu(\mathrm{Cl}(n))*\mu(K_Z)*\mu(\mathrm{Cl}(n)))
    \leq
    g(\nu_{\mathrm{PR}})\leq 15/16.
\end{equation*}

Next, let $K_2
:=
\big\{
    U\otimes I_{n-1}:
    U\in\SU(2)
\big\}$.
Since
$K_Z\subseteq K_2$,
every \(K_2\)-invariant vector is also \(K_Z\)-invariant. 
Therefore $M_{\rho}(\mu(K_2))\preceq M_{\rho}(\mu(K_Z))$.
Thus,
\begin{equation}
    g(\mu(\mathrm{Cl}(n))\ast\mu(K_2)\ast \mu(\mathrm{Cl}(n)))
    \leq
      g(\mu(\mathrm{Cl}(n))\ast\mu(K_Z)\ast \mu(\mathrm{Cl}(n)))\leq 15/16. 
\end{equation}

We now replace the two subgroup averages $M_\rho(\mu(\mathrm{Cl}(n))$ and $M_\rho(\mu(K_2))$ by finite walks. Write $H_1:=H\otimes I_{n-1}$ and $T_1:=T\otimes I_{n-1}$. 
The set \(\mathsf W_n\) generates \(\mathrm{Cl}(n)\) and
\(\{\ii H_1,T_1,T_1^{\dagger}\}\) generates a dense subgroup of \(K_2\); we adjoin
\(I\) to the latter set to make the walk lazy.
By \Cref{lemma:Clifford-expander} and the gap from~\cite{sarnak2015letter}, there are absolute constants
\(q_{\mathrm C},q_2<1\) such that, for every representation \(\rho\),
\begin{equation}
    \norm*{
        \underbrace{M_\rho(\mu(\mathsf W_n))}_{Q_{\mathrm C}}
        -\underbrace{M_\rho(\mu(\mathrm{Cl}(n)))}_{P_{\mathrm C}}
    }
    \leq q_{\mathrm C},
    \qquad
    \norm*{
        \underbrace{M_\rho(\mu(\{I,\ii H_1,T_1\}))}_{Q_2}
        -\underbrace{M_\rho(\mu(K_2))}_{P_2}
    }
    \leq q_2.
    \label{eq:finite-walk-gaps}
\end{equation}
Haar invariance gives
\begin{equation}
    Q_{\mathrm C}P_{\mathrm C}=P_{\mathrm C}Q_{\mathrm C}=P_{\mathrm C},
    \qquad
    Q_2P_2=P_2Q_2=P_2.
\end{equation}
Hence \(Q^\ell-P=(Q-P)^\ell\) for either pair \((Q,P)\). Thus,
for integers \(\ell_{\mathrm C},\ell_2\geq1\),
\cref{eq:finite-walk-gaps} implies
\begin{equation}
    \norm*{Q_{\mathrm C}^{\ell_{\mathrm C}}-P_{\mathrm C}}
    \leq q_{\mathrm C}^{\ell_{\mathrm C}},
    \qquad
    \norm*{Q_2^{\ell_2}-P_2}
    \leq q_2^{\ell_2}.
    \label{eq:powered-walk-gaps}
\end{equation}

Define
\begin{equation}
    \nu
    :=
    \mu(\mathsf W_n)^{*\ell_{\mathrm C}}
    *\mu(\{I,\ii H_1,T_1\})^{*\ell_2}
    *\mu(\mathsf W_n)^{*\ell_{\mathrm C}}.
\end{equation}
By \cref{eq:moment-convolution,eq:powered-walk-gaps} and a telescoping argument,
\begin{equation}
    \left\|M_\rho(\nu)-P_{\mathrm C}P_2P_{\mathrm C}\right\|
    \leq
    2q_{\mathrm C}^{\ell_{\mathrm C}}+q_2^{\ell_2}.
\end{equation}
For representations without invariant vectors, the preceding subgroup
estimate gives \(\|P_{\mathrm C}P_2P_{\mathrm C}\|\leq15/16\). Hence, by choosing sufficiently large absolute constant integers \(\ell_{\mathrm C}\) and \(\ell_2\), we have
\begin{equation}
    g(\nu)
    \leq
    \frac{15}{16}
    +2q_{\mathrm C}^{\ell_{\mathrm C}}
    +q_2^{\ell_2}
    \leq \frac{32}{33}. 
    \label{eq:finite-walk-gap}
\end{equation}

The measure \(\nu\) is supported on a constant-size family of
constant-depth circuits. By \cref{eq:finite-walk-gap}, its support has a
Kazhdan constant bounded below by an absolute constant; indeed, for every
unit vector \(v\) in a representation without invariant vectors,
\(\max_{U\in\operatorname{supp}\nu}\|\rho(U)v-v\|^2
\geq 2(1-g(\nu))\). Moreover,
\(\ii H_1=(H_1S_1)^6H_1\) and
\(I,H_1,S_1\in\mathsf W_n\), so every circuit in the support is
a bounded-length word in
\begin{equation}
    \mathsf V_n
    :=\mathsf W_n\cup\{T_1,T_1^{\dagger}\}
    =\mathsf W_n\cup\{H_1,T_1,T_1^{\dagger}\}.
\end{equation}
Then \Cref{lemma:short-product} transfers the uniform Kazhdan constant to \(\mathsf V_n\). Since $\mathsf V_n$ contains the identity and is closed under taking the inverses, \Cref{lem:kazhdan_gaps_relation} yields a uniform spectral gap for \(\mu(\mathsf V_n)\). 
Every generator in \(\mathsf V_n\) has depth one, and its only non-Clifford gates are \(T_1,T_1^{\dagger}\).
\end{proof}

We can now prove~\cref{cor:unitary-expander} using the KAK decomposition of the special unitary group. 
More precisely, we use the following lemma: 
\begin{lemma}\label{lemma:KAK-decomp}
For any $d\geq 8$, let $s\geq 1$ be the largest integer such that $2^{3s}\leq d$.
Then, there is a family of injective homomorphisms $\{\Lambda_{j}:\SU(2^{3s})\to \SU(d)\}_{j=1}^{216}$ such that any $U\in \SU(d)$ can be written as $U=\prod_{j=1}^{216}U_j$, with $U_j\in \Lambda_j(\SU(2^{3s}))$.
\end{lemma}
\begin{proof}
    First, we use the KAK (or Cartan) decomposition~\cite{wierichs2024kak} of $\rm{SU}(d)$ of the form $U=K_1 A K_2$. 
 We have $K_1,K_2\in \rm{S}(\rm{U}(d-k)\times \rm{U}(k))$ and $A$ can be obtained from embedding of $\rm \SU(2)^{\times k}$.
We first show that there are six embeddings of $\rm \SU(k)$ whose product is the full group $\rm \SU(d)$ if $k\leq d \leq 2m$. 
Indeed, write $d=k+r$ with $1\leq r\leq k$ and fix $\mathbb{C}^d=E\oplus P\oplus Q$, with $\dim E=k-r$, $\dim P=\dim Q=r$. 
Then, we can set $B=E\oplus P$ and $C=E\oplus Q$. 
The KAK decomposition than yields a decomposition
\begin{equation}
U=(L_B\oplus L_Q) A (R_B\oplus R_Q),
\end{equation}
where $A=I_E\oplus \bigoplus_{i=1}^r N_i$, $N_i\in \rm \SU(2)$. 
Set $l_B=\det L_B$, $l_Q=\det L_Q$, $r_B=\det R_B$, $r_Q=\det R_q$. 
As $\det U=1$ we have $l_B l_Q r_B r_Q=1$. 
We can now absorb the superfluous determinants of $L_B, L_Q, R_B$, and $R_Q$ into an $\SU(2)$ block: 
\begin{equation}
\hat{N}_i:=\begin{pmatrix}l_B&0\\0 &l_Q\end{pmatrix} A \begin{pmatrix}r_B&0\\0 &r_Q\end{pmatrix}.
\end{equation}
The corresponding block is then still in $\SU(2)$. 
This allows us to rescale $L_B, L_Q, R_B$, and $R_Q$ in a single direction to make each block's determinant $1$. 
In the case $r=1$ there is no $\SU(2)$ block to absorb the determinants, but the statement is trivial then.

Finally, set $m=2^{3s}$ and $m\leq d\leq 8m$. 
Then, the above condition $k\leq d\leq 2k$ is satisfied for some $k=m, 2m, 4m$. 
We simply apply the above decomposition recursively at most three times. 
Overall, we need at most $6^3=216$ embeddings.
 \end{proof}

\begin{proof}[Proof of~\cref{cor:unitary-expander}]
    The statement follows directly from~\cref{lemma:KAK-decomp} in combination with \cref{lemma:kazhdan-by-subgroups}:
    As $\mathcal{K}(\rm{SU}(d); \rm{SU}(d))\geq \sqrt{2}$, we can directly apply the lemma to obtain a uniformly bounded Kazhdan constant for $\bigcup_j \Lambda_{j}(\mathsf{V_{3s}})$, where $s$ is the largest integer such that $2^{3s}\leq d$. 
    Again, we can translate the Kazhdan constant to a spectral gap for $d\geq 8$ via~\cref{lem:kazhdan_gaps_relation}.
 The overall number of elements in the generator set is then at most $216\times 25=5400$.
 For $d<8$ we can always pick any sufficiently small generator set with a spectral gap, which exists by~\cite{bourgain2012spectral}.
\end{proof}

\section{Application 1: tightness of gap vs entanglement in 1D}\label{sec:vollaw}
Given a Hamiltonian $H$ on a 1D line of $n$ qubits with geometrically local interactions whose unique ground state is $\ket{\psi}$, the bipartite entanglement entropy (denote the two regions as $A$ and $B$) obeys
\begin{equation}
S(\psi_A) \leq O\left(\frac{\log^3 d}{\Delta}\right),
\end{equation}
where $d$ is the local dimension and $\Delta$ is the spectral gap~\cite{arad2013area}.  For the rest of the section, we will set $d = O(1)$ and ignore the dependence on it.  An open question has been whether linear scaling in $1/\Delta$ is optimal.  Phrased another way, what is the largest possible gap for the ground state of a 1D Hamiltonian to have $\Theta(n)$ bipartite entanglement?  If the largest gap is $\sim 1/n$, then the bound in Ref.~\cite{arad2013area} is tight.  Currently, the largest known gap scaling is $\sim 1/n^4$, provided in a construction in Ref.~\cite{gottesman2010entanglement}.  The Hamiltonian is frustrated and also has a somewhat undesirable feature of some of the Hamiltonian terms having norm $1/\text{poly}(n)$\footnote{While this is not mathematically problematic, physicists do not associate such a property with a natural Hamiltonian since the interaction strength or the strength of an external field depends on the system size and vanishes in the thermodynamic limit.   For this reason, we would like to avoid such Hamiltonian families.}.  

When the Hamiltonian is frustration-free, it is believed that the gap versus entanglement scaling obeys $S(\psi_A) \leq O\left(\frac{1}{\sqrt{\Delta}}\right)$ (assuming a constant local dimension).  Thus, if the ground state of a 1D Hamiltonian has $\Theta(L)$ bipartite entanglement, then the gap can be at most $\sim 1/L^2$.  In this section, we use the constant-depth quantum expander to prove the following theorem:
\begin{theorem}\label{thm:hamfamily}
There exists a family of frustration-free Hamiltonians $\{H_n: n \in \mathbb{Z}^+\}$ where $H_n$ is defined on a line of $n$ sites with $O(1)$ local dimension, such that:
\begin{enumerate}
\item $H_n$ has a unique ground state with bipartite entanglement entropy $\Theta(n)$
\item $H_n$ has a spectral gap $\Delta_n = \Theta(n^{-2})$.
\end{enumerate}
\end{theorem}
\subsection{Construction of Hamiltonian family}

We provide a construction of $H_n$.  $H_n$ is defined on a line of $\Theta(n)$ qudits whose local Hilbert space will be denoted as
\begin{equation}
\mathcal{H}_{\text{loc}} \cong \mathrm{span}\{\ket{\ast}, \ket{0}, \ket{1}, \ldots, \ket{m-1}, \ket{m}\}
\end{equation}
where $m = 17$ is the number of unitaries specifying the quantum expander plus one.  The extra unitary represented by $\ket{1}$ is the identity.  Next, we will choose the geometry below.
\begin{center}
\includegraphics[scale=0.65]{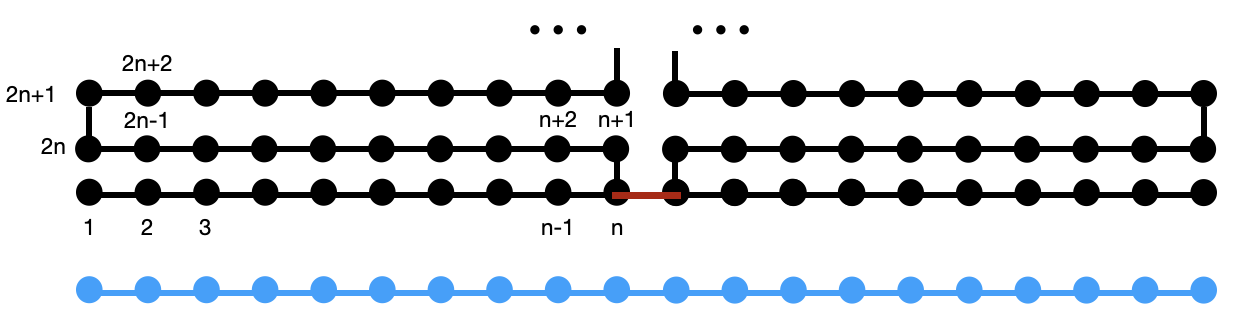}
\end{center}

The coloring specifies the different roles of qudits.  Blue qudits are \emph{data} qudits and these are responsible for generating large ground state entanglement.  Black qudits are \emph{clock} qudits, which are responsible to applying the quantum expander unitaries to the blue qudits.  Black qudits are divided into two halves, which we call left and right.  Edges correspond to Hamiltonian interactions.  The red link is a special interaction coupling left and right clock qudits. 

Note that the clock qudits form a ``snaking geometry''.  The total length of the snake is $\alpha n$ with $\alpha > 1$ for both left and right clock qudits.  This means that the Hamiltonian is still 1D local; one can see this by ``flattening'' the Hamiltonian into one that is defined on a line of $2n$ qudits, each with a larger but still constant local dimension.  However, we will work directly with the snaking geometry for ease of presentation.

The full Hamiltonian will be written in the form
\begin{equation}
H = H_{\mathrm{check}} + H_M + H_{\mathrm{prop}}.
\end{equation}
We will construct each of these in the following subsections.

\subsubsection{Constructing \(H_{\mathrm{check}}\)}

Denote the set of left and right clock qudits by $L$ and $R$.  In the diagram above, a canonical ordering is given for the qudits in $L$.  A similar ordering can be given for $R$.  We will be using these orderings moving forward.

The check Hamiltonian will be a diagonal Hamiltonian that enforces certain configurations of the clock qudits in the groundspace.  Define a word $w \in \{\ast, 0, 1, \cdots, m\}^{\alpha n}$.  Define a word to be valid in $L$ if it takes the form
\begin{equation}
w = \begin{cases}
\ast^k \# \, 0^{\alpha n - k -1} &  0\leq k \leq n-2 \\
\ast^{n-1} \# \ast^{k-n+1} \, 0^{\alpha n - k-1} & k \geq n-1
\end{cases}
\end{equation}
where the ordering of qubits is the canonical one and $\# \in \{1,2,\ldots,m\}$.  We want the ground space on both the left and right clock qudits to only consist of valid words.  $H_{\mathrm{check}}$ will energetically enforce this constraint.  We write $H_{\mathrm{check}} = H_{\mathrm{check}, L} + H_{\mathrm{check}, R}$.  $H_{\mathrm{check}, L}$ includes the following terms:
\begin{enumerate}
\item A term of the form $H_{1} = \sum_{k > n} \sum_{i=1}^m \ketbra{i}{i}_{k,L}$.  This enforces that the ground state subspace cannot have any $\#$ characters past position $n$.
\item A term of the form 
\begin{align}
H_2 = &\sum_{k < n} \sum_{i=1}^m \paren*{\ketbra{0}{0}_{k,L}\ketbra{i}{i}_{k+1,L} + \ketbra{i}{i}_{k,L}\ketbra{\ast}{\ast}_{k+1,L}} \nonumber\\&+ \sum_{k < n} \paren*{\ketbra{0}{0}_{k,L}\ketbra{\ast}{\ast}_{k+1,L} + \ketbra{\ast}{\ast}_{k,L}\ketbra{0}{0}_{k+1,L}} \nonumber\\&+ \sum_{k < n} \sum_{i,j=1}^m \ketbra{i}{i}_{k,L}\ketbra{j}{j}_{k+1,L}.
\end{align}
This ensures that the length-$n$ prefix of the word must take the form $\ast^a \# 0^b$ where $a + b + 1 = n$.  However, it also allows $\ast^n$ and $0^n$.
\item A term of the form $H_3 = \ketbra{0}{0}_{1,L} + \ketbra{\ast}{\ast}_{n,L}$.  This eliminates the words $\ast^n$ and $0^n$ from the ground state subspace.
\item A term of the form $H_4 = \sum_{k \geq n}\ketbra{0}{0}_{k,L} \ketbra{\ast}{\ast}_{k+1,L}$.  This enforces that the word following the length-$n$ prefix takes the form $\ast^c \, 0^d$ for $c+d = (\alpha -1) n$.
\end{enumerate}

The check Hamiltonian $H_{\mathrm{check}, L} = H_1 + H_2 + H_3+ H_4$.  Similarly, $H_{\mathrm{check}, R}$ can be constructed  by swapping $L$ and $R$ in $H_1, \ldots, H_4$.

\subsubsection{Constructing \(H_{\mathrm{prop}}\)}

Next, we discuss how to construct the propagation part of the Hamiltonian, which is responsible for applying unitaries to the data qubits.  We first discuss the Hamiltonian without the unitaries before discussing how to incorporate them.

On the left clock qudits, define the propagation Hamiltonian
\begin{equation}
H_{\mathrm{prop}, L,1} = -\ketbra{+}{+}_1-\ketbra{\ast}{\ast}_1 + \sum_{i=1}^m \sum_{k < n} (\ket{i}_{k}\ket{0}_{k+1} - \ket{\ast}_k\ket{i}_{k+1})(\bra{i}_k\bra{0}_{k+1} - \bra{\ast}_k\bra{i}_{k+1})
\end{equation}
where we have dropped the subscript $L$ in the indices.  Here, $\ket{+} = \frac{1}{\sqrt{m}}\sum_{i=1}^m \ket{i}$ and the first two terms enforce that $\ket{+}$ (and not some other superposition) propagates across the chain in the ground state.  For the remaining $(\alpha - 1)n$ sites, define
\begin{align}
h_{\mathrm{start}, 2} &= \sum_{i=1}^m(\ket{i}_n\ket{0}_{n+1}\ket{0}_{n+2} - \ket{i}_n\ket{\ast}_{n+1}\ket{0}_{n+2})(\bra{i}_n\bra{0}_{n+1}\bra{0}_{n+2} - \bra{i}_n\bra{\ast}_{n+1}\bra{0}_{n+2}) \\
h_{j, 2} &= (\ket{\ast}_{n+j}\ket{0}\ket{0} - \ket{\ast}_{n+j}\ket{\ast}\ket{0})(\bra{\ast}_{n+j}\bra{0}\bra{0} - \bra{\ast}_{n+j}\bra{\ast}\bra{0}) \\
h_{\mathrm{end}, 2} &= (\ket{\ast}_{\alpha n -1}\ket{0}_{\alpha n} - \ket{\ast}_{\alpha n -1}\ket{\ast}_{\alpha n})(\bra{\ast}_{\alpha n -1}\bra{0}_{\alpha n} - \bra{\ast}_{\alpha n -1}\bra{\ast}_{\alpha n})
\end{align}
where the notation $\ket{\ast}_{n+j}\ket{0}\ket{0}$ is shorthand for $\ket{\ast}_{n+j}\ket{0}_{n+j+1}\ket{0}_{n+j+2}$.  Then, define the Hamiltonian on the remaining $(\alpha - 1)n$ sites via
\begin{align}
H_{\mathrm{prop}, L,2} = h_{\mathrm{start}, 2} + \sum_{j=1}^{(\alpha - 1) n - 2} h_{j,2} + h_{\mathrm{end}, 2}.
\end{align}
The first term creates a domain wall of $\ast$ only if the character at site $n$ is $\#$.  The remaining terms propagate a domain wall of $\ast$ along the remaining $(\alpha - 1) n$ sites.    We can define $H_{\mathrm{prop}, L} = H_{\mathrm{prop}, L,1} + H_{\mathrm{prop}, L,2}$ and similarly define $H_{\mathrm{prop}, R}$.

Write $H_{\mathrm{prop}, L} = h_1 + \sum_{k<n} h_{k, k+1} + \sum_{k\geq n} h_{k, k+1, k+2} + h_{\alpha n-1, \alpha n}$ with $h_{i,i+1,i+2}$ including all terms supported completely in sites $i$, $i+1$ and $i+2$, $h_{i,i+1}$ including all terms supported entirely in sites $i$ and $i+1$, $h_1 = -\ketbra{+}{+}_1-\ketbra{\ast}{\ast}_1$, and $h_{\alpha n-1, \alpha n} = h_{\mathrm{end},2}$. Then, this Hamiltonian is frustration-free and the ground state is the simultaneous ground state of $\{h_1, h_{1,2}, \cdots, h_{n, n+1, n+2}, \cdots, h_{\alpha n-1, \alpha n}\}$. Additionally, $H_{\mathrm{check}, L} + H_{\mathrm{prop}, L}$ is similarly frustration-free and the ground state restricted to the left clock qudits is unique and given by
\begin{equation}\label{gsleft}
\ket{\mathrm{GS}_L} = \frac{1}{\sqrt{\alpha n}}\sum_{k \leq n-2} \ket{\ast^k}\otimes \ket{+}\otimes\ket{0^{\alpha n - k-1}} + \frac{1}{\sqrt{\alpha n}} \sum_{k > n-2} \ket{\ast^{n-1}}\otimes \ket{+} \otimes\ket{\ast^{k-n+1} \, 0^{\alpha n-k-1}}
\end{equation}
where the notation $\ket{\ast^k} = \ket{\ast}^{\otimes k}$ and so on.  The state $\ket{\mathrm{GS}_R}$ has a similar form.

Next, we incorporate the data qubits.  Recall that there are $2n$ data qubits and these can be divided into left and right halves.  This will modify $H_{\mathrm{prop}, L,1}$ and $H_{\mathrm{prop}, R,1}$ into $H'_{\mathrm{prop}, L,1}$ and $H'_{\mathrm{prop}, R,1}$.  Recall that $m=17$ is one more than the number of unitaries specifying the quantum expander.  Call each of these unitaries $U_i$, with $i \in \{2,3,\ldots,m\}$.  The unitary $U_1 = I$.  
These unitaries can be decomposed as
\begin{equation}
U_i = \prod_{k=1}^{n-1} U_{i,k}
\end{equation}
where $k$ indexes a qubit on the left half of the data qubits that $U_{i,k}$ starts at, and each of the $U_{i,k}$ is assumed to have constant range.  This results in a ``staircase'', or ``sequential'' circuit for $U_i$.  Since in our case $U_i$ is depth-$1$ for all $i$, such a decomposition is possible.  We then define
\begin{equation}
\begin{aligned}
H'_{\mathrm{prop}, L,1}
={}&-\ketbra{+}{+}_1-\ketbra{\ast}{\ast}_1 \\
&+\sum_{i=1}^m \sum_{k < n}
\bigl(\ket{i}_{k}\ket{0}_{k+1} \otimes I
    - \ket{\ast}_k\ket{i}_{k+1} \otimes U_{i,k}\bigr)
\bigl(\bra{i}_k\bra{0}_{k+1} \otimes I
    - \bra{\ast}_k\bra{i}_{k+1}\otimes U^{\dagger}_{i,k}\bigr).
\end{aligned}
\end{equation}
where the notation $A \otimes B$ indicates that $A$ acts on the left clock qudits and $B$ acts on the left data qudits.  This term is designed to implement a standard Feynman-Kitaev Hamiltonian so that if the state of the clock qudit $k$ is $\ket{i}_k$, the unitary $U_{i,k}$ is applied to the data.  Similarly, we can write
\begin{equation}
\begin{aligned}
H'_{\mathrm{prop},R,1}
={}&-\ketbra{+}{+}_1-\ketbra{\ast}{\ast}_1 \\
&+\sum_{i=1}^m \sum_{k < n}
\bigl(\ket{i}_{k}\ket{0}_{k+1} \otimes I
    - \ket{\ast}_k\ket{i}_{k+1} \otimes \overline{U}_{i,k}\bigr)
\bigl(\bra{i}_k\bra{0}_{k+1} \otimes I
    - \bra{\ast}_k\bra{i}_{k+1}\otimes U^T_{i,k}\bigr).
\end{aligned}
\end{equation}
which implements the complex conjugate of $U_i$.  The full propagation Hamiltonian is $H_{\mathrm{prop}} = H'_{\mathrm{prop}, L,1} + H_{\mathrm{prop}, L,2} + H'_{\mathrm{prop}, R,1} + H_{\mathrm{prop}, R,2}$.

Finally, we note that the unitaries in our depth-$1$ quantum expander are defined on a circle rather than a line.  This can be accommodated for in the Hamiltonian without sacrificing 1D locality.  For this, we simply fold the left and right chains in half, so that we can apply a unitary between the first and last qubits in a geometrically local way.  1D locality is still obeyed, and the bipartite entanglement entropy about the middle is unaffected.  We may also use the depth-$1$ expander in Ref.~\cite{liu2026almost}, which is defined on a line and does not require the extra folding trick.  This comes at the expense of increasing the local Hilbert space dimension by quite a bit, given that their expander requires more unitaries than ours.

\subsubsection{Constructing $H_M$}

Finally, we must specify the interaction terms corresponding to the red link in the diagram.  Inspired by the construction of Ref.~\cite{aharonov2014local}, we will choose it to take the form 
\begin{equation}
H_{M} = \frac{1}{2}\sum_{i = 2}^m (\ket{1}_{n,L}\ket{i}_{n,R} - \ket{i}_{n,L}\ket{1}_{n,R})(\bra{1}_{n,L}\bra{i}_{n,R} - \bra{i}_{n,L}\bra{1}_{n,R})
\end{equation}
This imposes a constraint on amplitudes in the ground state that will force the left and right data qubits to be very entangled with each other.  Note that $\norm{H_M} \leq 1$.

\subsection{Proof of Theorem~\ref{thm:hamfamily}}

We now provide a proof of Theorem~\ref{thm:hamfamily}.  The proof can be divided into a few parts: first, we construct the ground state of $H_n$ (for succinctness, we will drop the subscript $n$) and argue that it has entanglement entropy $\Theta(n)$.  Then, we provide a gap estimate for $H_L$ and $H_R$.  Then, we prove a gap estimate when adding $H_M$.

Recall that by construction, the ground state of $H_L = H_{\mathrm{check}, L} + H_{\mathrm{prop}, L}$ is given in Eqn.~\ref{gsleft}.  Define the unitary
\begin{equation}
W_{i,j} = \prod_{k=1}^j U_{i,k}.
\end{equation}
Then, it follows that the ground state of $H_{\mathrm{check}, L} + H'_{\mathrm{prop}, L, 1} + H_{\mathrm{prop}, L, 2}$ is
\begin{equation}
\begin{aligned}
\ket{\mathrm{GS}_L'}
={}& \frac{1}{\sqrt{m\alpha n}}
    \sum_{a=1}^m\sum_{k \leq n-2}
    \ket{\ast^k a\, 0^{\alpha n-k-1}}
    \otimes W_{a,k}\ket{\phi} \\
&+ \frac{1}{\sqrt{m\alpha n}}
    \sum_{a=1}^m\sum_{k > n-2}
    \ket{\ast^{n-1}a\ast^{k-n+1}\,0^{\alpha n-k-1}}
    \otimes U_a\ket{\phi}.
\end{aligned}
\end{equation}
Define $\ket{\alpha_{a,k}} = \ket{\ast^k a\, 0^{\alpha n - k-1}}$ for $k \leq n-2$ and define $\ket{\beta_{a,k}} = \ket{\ast^{n-1} a \ast^{k-n+1} \, 0^{\alpha n-k-1}}$ for $k > n-2$.  Define the unnormalized state 
\begin{equation}
    \ket{\beta_{a}} =\sum_{k > n-2} \ket{\ast^{n-1} a \ast^{k-n+1} \, 0^{\alpha n-k-1}}.
\end{equation}
We may then write
\begin{equation}
\ket{\mathrm{GS}_L'} = \frac{1}{\sqrt{m\alpha n}}\sum_{a=1}^m\sum_{k \leq n-2} \ket{\alpha_{a,k}} \otimes W_{a,k} \ket{\phi} + \frac{1}{\sqrt{m\alpha n}} \sum_{a=1}^m \ket{\beta_{a}}\otimes U_{a} \ket{\phi}.
\end{equation}
A similar ground state can be written for $H_{\mathrm{check}, R} + H'_{\mathrm{prop}, R, 1} + H_{\mathrm{prop}, R, 2}$ but with $W_{i,j}$ replaced with $\overline{W}_{i,j}$ and $U_i$ replaced with $\overline{U}_i$; call this $\ket{\mathrm{GS}_R'}$.  Note that $\ket{\phi}$ is an arbitrary state on the data qubits; thus, the ground space has exponentially large dimension.

Next, we discuss the ground state upon adding $H_M$.  An arbitrary state in the ground space of $H-H_M$ can be written as
\begin{align}\label{eq:genstateinker}
\ket{\psi} = \frac{1}{m\alpha n}&\sum_{a,b=1}^m\sum_{k_L,k_R \leq n-2} \ket{\alpha_{a,k_L}}_L \ket{\alpha_{b,k_R}}_R \otimes (W_{a,k_L} \otimes \overline{W}_{b,k_R}) \ket{\phi} \nonumber \\ &+ \frac{1}{m\alpha n} \sum_{a,b=1}^m \sum_{k_R \leq n-2}\ket{\beta_{a}}_L \ket{\alpha_{b, k_R}}_R\otimes (U_{a} \otimes \overline{W}_{b, k_R}) \ket{\phi} \nonumber \\
&+ \frac{1}{m\alpha n} \sum_{a,b=1}^m \sum_{k_L \leq n-2}\ket{\alpha_{a, k_L}}_L \ket{\beta_{b}}_R \otimes (W_{a, k_L}\otimes \overline{U}_{b}) \ket{\phi} \nonumber \\
&+ \frac{1}{m\alpha n}\sum_{a,b=1}^m \ket{\beta_{a}}_L \ket{\beta_{b}}_R \otimes (U_{a} \otimes \overline{U}_{b}) \ket{\phi}.
\end{align}
Since $H_M \geq 0$, the ground state of $H$ will satisfy $H_M \ket{\psi} = 0$.  Evaluating $H_M \ket{\psi}$ eliminates the first three terms; requiring that it annihilates the last term leaves us with the system of equations
\begin{equation}
(U_1 \otimes \overline{U}_b) \ket{\phi} = (U_b \otimes \overline{U}_1) \ket{\phi}
\end{equation}
for all $b$.  If $\{U_b: b \geq 2\}$ form a quantum expander and $U_1 = I$, then $\ket{\phi} = \ket{\Gamma}$ is the only solution, where $\ket{\Gamma}$ is a state of $n$ Bell pairs between the left and right regions.  This follows a similar argument to Ref.~\cite{aharonov2014local}.  This gives us the ground state of $H$, which we denote by $\ket{\mathrm{GS}}$, and it is unique.
\begin{lemma}\label{lem:highent}
The ground state $\ket{\mathrm{GS}}$ of $H$ has bipartite entanglement entropy $\Theta(n)$.
\end{lemma}
\begin{proof}
Define the unitary
\begin{equation}
U_L = \sum_{a = 1}^m \sum_{k_L \leq n-2} \ketbra{\alpha_{a,k_L}}{\alpha_{a,k_L}}_L \otimes W^\dagger_{a,k_L} + \sum_{a = 1}^m \sum_{k_L > n-2} \ketbra{\beta_{a,k_L}}{\beta_{a,k_L}}_L \otimes U^\dagger_{a} + \sum_{\gamma \, \mathrm{invalid}} \ketbra{\gamma}{\gamma}_L \otimes I
\end{equation}
where $\gamma$ is an invalid word and the last term is needed for $U_L$ to be unitary.  Similarly, define
\begin{equation}
U_R = \sum_{a = 1}^m \sum_{k_R \leq n-2} \ketbra{\alpha_{a,k_R}}{\alpha_{a,k_R}}_R \otimes W^T_{a,k_R} + \sum_{a = 1}^m \sum_{k_R > n-2} \ketbra{\beta_{a,k_R}}{\beta_{a,k_R}}_R \otimes U^T_{a} + \sum_{\gamma \, \mathrm{invalid}} \ketbra{\gamma}{\gamma}_R \otimes I.
\end{equation}
Note that $(U_L \otimes U_R) \ket{\mathrm{GS}}$ takes the form $\ket{\xi}_L \otimes \ket{\xi}_R \otimes \ket{\Gamma}$ for unspecified $\ket{\xi}$.  Then, calling $\rho = \ketbra{\mathrm{GS}}{\mathrm{GS}}$ and $\Tr_{R}(\cdot)$ the partial trace over $R$ and the right half of the data qubits, we have $\Tr_{R}((U_L \otimes U_R)\rho(U_L \otimes U_R)^\dagger) = U_L (\Tr_{R}\rho) U_L^\dagger$  and thus the bipartite entanglement entropy of $(U_L \otimes U_R) \ket{\mathrm{GS}}$ and $\ket{\mathrm{GS}}$ are the same.  The bipartite entanglement entropy of $(U_L \otimes U_R) \ket{\mathrm{GS}}$ is the bipartite entanglement entropy of $\ket{\Gamma}$, which is $n$.  This proves the Lemma. 
\end{proof}

Next, we prove the gap estimate.  First, we establish the gap for the Hamiltonian $H - H_M$.  Note that $H - H_M$ has a degenerate ground space so the gap corresponds to $\mathrm{gap}(H-H_M) = \min_{E \neq E_0} E-E_0$ where $E_0$ is the ground state energy. 
\begin{lemma}\label{lem:H-HMgap}
The gap of $H - H_M$ is $\Theta(n^{-2})$.
\end{lemma}
\begin{proof}
We first observe the fact that both Hamiltonians $H$ and $H-H_M$ have a block structure $H_{\mathrm{valid}} \oplus H_{\mathrm{invalid}}$, where $H_{\mathrm{valid}}$ is defined within the subspace of valid words and $H_{\mathrm{invalid}}$ is defined for invalid words.  Suppose $\ket{\mathrm{inv}}$ is a state in the invalid subspace and $\ket{\mathrm{val}}$ is a state in the valid subspace.  Then $\mel{\mathrm{inv}}{H_{\mathrm{check}}}{\mathrm{inv}} \geq \mel{\mathrm{val}}{H_{\mathrm{check}}}{\mathrm{val}} + 1$.  Thus, because $H$ is frustration-free, the ground state energies of both sectors satisfy $E_0(H_{\mathrm{invalid}}) \geq E_0(H_{\mathrm{valid}}) + 1$.  So we only need to focus on computing the gap of $H_{\mathrm{valid}}$.

Within $H_{\mathrm{valid}}$, the effect of $H_{\mathrm{check}}$ is adding a constant times the identity matrix, which does not modify the gap.  Next, by conjugating the Hamiltonian by $U_L \otimes U_R$ reduces the problem to studying the spectrum of $H_{\mathrm{prop}, L, 1} + H_{\mathrm{prop}, L, 2} + H_{\mathrm{prop}, R, 1} + H_{\mathrm{prop}, R, 2}$.  By symmetry, we only need to study $H_{\mathrm{prop}, L, 1} + H_{\mathrm{prop}, L, 2}$.  Define the states 
\begin{equation}
\ket{i, k} = \begin{cases}
\frac{1}{\sqrt{m}}\sum_{j=1}^m \omega^{kj}\ket{\alpha_{j, i}} & i \leq n-2 \\
\frac{1}{\sqrt{m}}\sum_{j=1}^m \omega^{kj}\ket{\beta_{j, i}} & i > n-2
\end{cases}
\end{equation}
for $k \in \{0,1,\ldots, m-1\}$ and $\omega = \exp(2\pi i/m)$.  These states satisfy orthonormality $\braket{i,k}{i',k'} = \delta_{i,i'} \delta_{k,k'}$.  By direct computation, the Hamiltonian further block diagonalizes as $H_{\mathrm{valid}} = \bigoplus_{k=0}^{m-1} H_k$, where $H_k$ is projected onto the subspace of states spanned by $\ket{i,k}$ for all $i$.  When $k > 0$, the Hamiltonian $H_k$ is the tridiagonal matrix
\begin{equation}
H_k = \begin{pmatrix} 
    1 & -1 &  & \\
    -1 & 1 & -1 &  \\
     & -1 & 1 & -1 &  \\
     &  & -1 & \ddots & \ddots & \\ 
     &  &  & \ddots\\
    \end{pmatrix}
\end{equation}
When $k = 0$, the Hamiltonian $H_0$ is the tridiagonal matrix
\begin{equation}\label{eq:k=0specialization}
H_0 = \begin{pmatrix} 
    0 & -1 &  & \\
    -1 & 1 & -1 &  \\
     & -1 & 1 & -1 &  \\
     &  & -1 & \ddots & \ddots & \\ 
     &  &  & \ddots\\
    \end{pmatrix}
\end{equation}
A standard computation shows that $E_0(H_{k > 0}) - E_0(H_0) \approx \pi^2/(4 \alpha^2 n^2)$ for large $n$ and that $\mathrm{gap}(H_0) \approx  \pi^2/(\alpha^2 n^2)$.  Here, $\approx$ means equality with error terms $\Theta(n^{-3})$.  Thus, $\mathrm{gap}(H - H_M) = \mathrm{gap}(H_{\mathrm{valid}}) \approx \pi^2/(4 \alpha^2 n^2)$, which is $\Theta(n^{-2})$.  When $\alpha = 2$, we can lower bound the gap by $4/(25 n^2)$\footnote{This follows from the gap of $H_{\mathrm{valid}}$ being $2 - 2 \cos \left(\frac{\pi}{2\alpha n + 1}\right)$.  With $\alpha = 2$, $2 - 2 \cos \left(\frac{\pi}{4 n + 1}\right) \geq \frac{4}{25 n^2}$}.
\end{proof}

The final step is to compute the gap when $H_M$ is added back in.  For this, we use the following lemma
\begin{lemma}\label{lem:proj}
Let $A\geq 0$ and $B \geq 0$ be operators, and define $P$ to be the projector onto $\ker A$.  Suppose there exists a projector $P_0 \subseteq \ker B \cap \ker A$ onto a 1-dimension subspace such that
\begin{equation}
A \geq b(1-P), \hspace{0.5cm} PBP \geq c(P-P_0)
\end{equation}
with $\norm{B} \leq 1$ and $b < c/2$.  Then, $A + B$ has ground space projector $P_0$ and gap $\geq \eta b$ where $\eta$ only depends on $c$.
\end{lemma}
\begin{proof}
Let $\ket{\psi}$ be a state satisfying $P_0 \ket{\psi} = 0$.  Write $\ket{\psi} = P\ket{\psi} + (I-P)\ket{\psi} \equiv \ket{x} + \ket{y}$.  Then,
\begin{align}
\mel{\psi}{A+B}{\psi} &\geq b \norm{y}^2 + \mel{\psi}{B}{\psi} \\
&\geq b \norm{y}^2 + \mel{x}{B}{x} - 2 \norm{B} \norm{x} \norm{y} \\
& \geq b \norm{y}^2 + c \mel{\psi}{P}{\psi} - c \mel{\psi}{P_0}{\psi} - 2 \norm{x} \norm{y} \\
& \geq b \norm{y}^2 + c \norm{x}^2 - 2 \norm{x} \norm{y}
\end{align}
Alternatively, we can also write
\begin{equation}
\mel{\psi}{A+B}{\psi} = \mel{\psi}{A}{\psi} + \mel{\psi}{B}{\psi} \geq \mel{\psi}{A}{\psi} \geq b \norm{y}^2.
\end{equation}
Suppose $\norm{y} \leq c \norm{x}/4$.  Then, we use the first inequality to write $\mel{\psi}{A+B}{\psi} \geq \min(b, c/2)(\norm{x}^2 + \norm{y}^2) = b$.  Suppose $\norm{y} \geq c \norm{x}/4$.  Then we can use the second inequality to write $\mel{\psi}{A+B}{\psi} \geq b \norm{y}^2 \geq b/2(\norm{y}^2 + c^2 \norm{x}^2/16) \geq b \min(1/2, c^2/32)$.
\end{proof}
Next, we will choose $\alpha = 2$ for the rest of the analysis.
\begin{lemma}\label{lem:HMgap}
With $A = H-H_M$ and $B = H_M$ where $A$ and $B$ are defined in Lemma~\ref{lem:proj} and $H$ is shifted by a suitable constant to have 0 ground state energy, $c = \Delta(m-1)/4m^2$ where $\Delta$ is the gap of the quantum expander.
\end{lemma}
\begin{proof}
We compute $P H_M P$, where $P$ is a projector onto $\ker(H-H_M)$.  Suppose $P_0$ is the projector onto the unique ground state of $H$; then any state $\ket{\psi} \in \supp(P-P_0)$ can be written as Eqn.~\ref{eq:genstateinker} with $\ket{\phi}$ satisfying $\braket{\phi}{\Gamma} = 0$.  Then, a computation gives (with $\alpha = 2$)
\begin{equation} 
\mel{\psi}{H_M}{\psi} = \frac{(n+1)^2}{8 m^2n^2}\sum_{i=2}^m \mel{\phi}{(U^\dagger_1 \otimes U^T_i - U^\dagger_i \otimes U^T_1)(U_1 \otimes \overline{U}_i - U_i \otimes \overline{U}_1)}{\phi}.
\end{equation}
To simplify this, we use an operator-state correspondence following Ref.~\cite{aharonov2014local}.  Then, $\ket{\Gamma} \to I/2^{n/2}$, and $\ket{\phi} \to \Phi$ with $\norm{\Phi}_2 = 1$ and $\Tr \Phi = 0$ following from $\braket{\Gamma}{\phi} = 0$.  Then, using the fact that $U_1 = I$, the previous equation becomes
\begin{align} 
\mel{\psi}{H_M}{\psi} &= \frac{(n+1)^2}{8 m^2 n^2}\sum_{i=2}^m \norm{U_i \Phi - \Phi U_i}_2^2 \\ &\geq \frac{1}{8 m^2}\sum_{i=2}^m \norm{U_i \Phi U_i^{\dagger} - \Phi}_2^2 \geq \frac{\Delta (m-1)}{4m^2}. 
\end{align}
where the last inequality holds from an identical analysis in~\cite{aharonov2014local}, Sec. 4.2.  This provides the value of $c$.  Recall that $\Delta = 1/(2\cdot10^8)$ from Theorem~\ref{thm:expander_gap}.
\end{proof}
We now arrive at the desired Theorem:
\begin{proof}[Proof of Theorem~\ref{thm:hamfamily}]
From Lemma~\ref{lem:proj} and Lemma~\ref{lem:H-HMgap} we see that $H-H_M \geq b(1-P)$ for $b = 4/(25 n^2)$.  Using Lemma~\ref{lem:HMgap}, $c = \Theta(1)$ and is independent of $n$.  Therefore, Lemma~\ref{lem:proj} tells us that the gap of $H$ is $\Omega(n^{-2})$ and Lemma~\ref{lem:highent} proves a unique ground state with entanglement $\Theta(n)$.

It remains to show an upper bound on the gap.  For this, we use a variational principle.  Recall the definition of $U_L$ and $U_R$ and define the excited state
\begin{equation}
\ket{e} = (U_L \otimes U_R)^\dagger \ket{E}_L \otimes \ket{\Omega}_R \otimes \ket{\Gamma_n}
\end{equation}
where $\ket{\Omega}_R$ is the ground state of $H_{\mathrm{prop}, R, 1} + H_{\mathrm{prop}, R, 2}$, $\ket{E}_L$ is the first excited state of $H_{\mathrm{prop}, L, 1} + H_{\mathrm{prop}, L, 2}$ in the $k=0$ block of the valid subspace (whose Hamiltonian is given in Eq.~\ref{eq:k=0specialization}), and $\ket{\Gamma_n}$ is the maximally entangled state on the data register.  Note that $\ket{E}_L$ has energy $\Theta(n^{-2})$ larger than the ground state energy of the Hamiltonian in Eq.~\ref{eq:k=0specialization}.   It is also clear that $\braket{e}{\mathrm{GS}} = 0$.  Then, $\mel{e}{H - H_M}{e} - \mel{\mathrm{GS}}{H - H_M}{\mathrm{GS}} = \Theta(n^{-2})$.  A simple calculation shows that $\mel{e}{H_M}{e} = 0$, so this state has gap $O(n^{-2})$ above the ground state energy, thus proving a gap $\Theta(n^{-2})$ when combined with the lower bound.
\end{proof}

Note that for the proof of the main theorem, we only needed the unitaries of the quantum expander to be sequential/staircase circuits.  This is a strictly larger class than constant-depth quantum expanders.  It would be interesting if there were simple constructions of sequential quantum expanders which may decrease the number of unitaries or provide a larger gap.

\paragraph{Mutual information of Gibbs states.}
One corollary is related to the area law scaling of mutual information in Gibbs states~\cite{wolf2008area, kuwahara2021improved}.  In particular, Ref.~\cite{kuwahara2021improved} showed that for a Gibbs state of Hamiltonian $H$ at inverse temperature $\beta$, denoted by $\rho_{\beta}$, the mutual information $I_\beta(A:B)$ scales like
\begin{equation}
I_{\beta}(A:B) \leq \beta^{\gamma} |\partial A|
\end{equation}
for $\gamma = 2/3$.  Thus, while $\beta$ is normally thought of as the thermal correlation length, for entanglement the correlation length is $\beta^{\gamma}$.  It is believed though that the optimal exponent $\gamma_c$ is $1/2$.  In fact, Ref.~\cite{kuwahara2021improved} provides the bound $1/5 \leq \gamma_c \leq 2/3$.  The lower bound follows from using the 1D construction from Ref.~\cite{gottesman2010entanglement}, in particular the fact that it has gap $1/n^4$, to show that $I_{\beta}(A:B)$ when $\beta = \Theta(n^5)$ is close to that of the ground state which is $\Theta(n)$.  Thus, it follows that $I_{\beta}(A:B) = \Theta(\beta^{1/5} |\partial A|)$ for this Hamiltonian and choice of $\beta$.  Using our 1D construction instead of Ref.~\cite{gottesman2010entanglement}, we find $I_{\beta}(A:B) = \Theta(\beta^{1/3} |\partial A|)$ for $\beta = \Theta(n^3)$, thus improving the bound on $\gamma_c$ to $1/3 \leq \gamma_c \leq 2/3$.

\section{Application 2: streaming testing of highly entangled states}\label{sec:streaming}

The second application of the quantum expander is to convert the communication protocol for testing $n$ EPR pairs in \cite{aharonov2014local} into a streaming protocol. In particular, we will test for closeness to the state
\begin{equation}
    \ket{\Gamma_n}
    = \frac{1}{2^{n/2}}
      \sum_{x\in\{0,1\}^n}
      \ket{x}\otimes\ket{x}.
\end{equation}
The streaming model we use is defined below.  

\begin{definition}[Streaming model]
In the streaming model, a state $\rho$ on $2n$ qubits is sent to a verifier in the order $1,2,\ldots,2n$. The verifier has access to a small memory storing $O(1)$ qubits and cbits.  At time $t$, when the verifier receives qubit $t$, they may apply a quantum operation to qubit $t$ and any data in memory and they may discard some of the qubits or cbits in memory.  We also assume the stream sends constant-sized classical messages marking times $1$, $n$, and $2n$. After receiving the last marker, the verifier outputs accept or reject.
\end{definition}

Let $\{U_1,\ldots,U_{16}\}$ be the depth-$1$ quantum expander from \Cref{sec:expander}, with gap $\Delta$.  Suppose that the range of the gates in the quantum expander is $\ell$.  Recall that the expander is defined on a 1D ring so that some of the CNOT gates will have support on qubits $1$ and $n$.  We call these \emph{wraparound gates}.

\begin{claim}\label{claim:streaming}
There is a one-pass streaming test for $\ket{\Gamma_n}$ using $O(\ell+\log m) = O(1)$ qubits and $O(1)$ classical bits of memory that implements the POVM $\{M, I-M\}$ with
\begin{equation}
\norm{M - \ketbra{\Gamma_n}{\Gamma_n}} \leq 1-\Delta
\end{equation}
where $\Delta$ is the gap of the depth-$1$ quantum expander.
\end{claim}

\begin{proof}
The verifier prepares $\ket{+}_A=\frac{1}{\sqrt{m}}\sum_{i=1}^m\ket{i}_A$ where $m=16$.
While the first half is streamed, they apply the controlled unitary $\sum_i \ketbra{i}{i}_A \otimes U_i$; after receiving the classical message marking time $n$, they apply the controlled unitary $\sum_i \ketbra{i}{i}_A \otimes \overline{U}_i$.  After receiving the classical message marking time $2n$, they measure $A$ with POVM $\{\ketbra{+}{+},I-\ketbra{+}{+}\}$ and accept on the first outcome.  This corresponds to the POVM $\{A^\dagger A,I-A^\dagger A\}$ on the qubits with
\begin{equation}
A = \frac{1}{m} \sum_{i=1}^m U_i \otimes \overline{U}_i.
\end{equation}

To implement the controlled unitary in the streaming setting, the verifier retains the first $\ell-1$ qubits and a rolling window containing its last seen $\ell-1$ qubits. Some gates will have support on qubits $1$ and $n$.  For those that do not, once the verifier has all qubits in their support, they apply the controlled gate. Once the middle marker is sent, the verifier, having stored the first $\ell-1$ qubits, applies the controlled wraparound gates and discards all qubits in memory. The verifier behaves analogously for the second half of the stream. 

Recall that the unitaries $U_7, \cdots, U_{16}$ of the depth-$1$ expander have a translation-invariance with period $3$ property.  All other unitaries have gates supported either on the first few qubits or the last few qubits. Due to this property, the verifier only needs an $O(1)$-sized look-up table to implement the unitaries of the depth-$1$ expander. 

The remaining analysis reduces to the expander EPR test of Ref.~\cite[Section~3.2]{aharonov2014local}, which proves the bound on the accept probability.  Finally, the verifier stores the $\lceil\log m\rceil$-qubit register $A$, at most $2(\ell-1)$ streamed qubits, and a constant number of classical control bits. This proves the claim.
\end{proof}
Finally, we prove the theorem advertised in the introduction:
\begin{theorem}
There exists a streaming state tester storing $O(\Delta^{-1}\log(1/\epsilon))$ qubits that implements a POVM $\{M, I-M\}$ with
\begin{equation}
\norm{M - \ketbra{\Gamma_n}{\Gamma_n}} \leq \epsilon.
\end{equation}
\end{theorem}
\begin{proof}
We repeat the quantum expander many times to amplify the gap $\Delta$.  More precisely, fix an integer $r\geq 1$ and replace the original set of unitaries by all products
\begin{equation}
    U_{i_r}\cdots U_{i_2}U_{i_1},
    \qquad i_1,\ldots,i_r\in[K].
\end{equation}
Construct a quantum expander with these unitaries.  Then, this expander still has $\ket{\Gamma_n}$ as a $+1$ eigenvalue state, while its norm on any orthogonal state is at most $(1-\Delta)^r$. Therefore, the generalization of the POVM in~\cref{claim:streaming} to the iterated expander gives
\begin{equation}
    \norm{A^\dagger A - \ketbra{\Gamma_n}{\Gamma_n}} \leq (1-\Delta)^r.
\end{equation}
Next, we analyze the space complexity.  Each unitary of the expander is now a depth-$r$ local circuit. It can still be implemented in one pass by keeping a larger window of streamed qubits: since each layer has range $\ell$, a sliding window of $O(r\ell)$ qubits is sufficient to apply all $r$ layers before those qubits are discarded. To handle the wraparound gates, the verifier retains the first $O(r\ell)$ qubits until receiving the middle or final marker. It also stores the $r$ labels $i_1,\ldots,i_r$, which requires $r\lceil\log m\rceil$ qubits, in all requiring $O(r\ell + r \log m)$ qubits.  Since $\ell, m = O(1)$ and setting $(1-\Delta)^{r} = \epsilon$, the space complexity is upper bounded by $O(\Delta^{-1} \log(1/\epsilon))$.
\end{proof}

Additionally, this streaming test has zero probability of a false negative: if the input state is $\ket{\Gamma_n}$ the test always outputs accept.

Finally, we briefly mention generalizations.  If the goal is to test $U\ket{\Gamma_n}$ where $U$ is a constant-depth (almost translation-invariant) circuit or an MPO of constant bond dimension, then it suffices to construct the operator $A' = U A U^{\dagger}$, which remains a sum of a constant number of matrix product operators.  This can be implemented in the streaming setting, which mildly generalizes the state tester beyond testing $\ket{\Gamma_n}$.

\section{Improved local gap scaling of 1D ground states}
\label{append:localgapscaling}

Our construction in \cref{sec:vollaw} raises the question of how it is possible for entanglement entropy to scale with inverse gap in frustration-free Hamiltonians.  Our example has $S \gtrsim \Delta_{\text{gbl}}^{-1/2}$, where we have introduced a slightly different notation, denoting the gap or ``global gap'' by $\Delta_{\text{gbl}}$.  Define also $\Delta_{\text{loc}}$ to be the local gap:
\begin{equation}
\Delta_{\text{loc}} := \min_{S \subseteq [n]} \mathrm{gap}\left(\sum_{i,i+1 \in S} h_{i,i+1}\right)
\end{equation}
where $\mathrm{gap}(H)$ is the smallest non-zero eigenvalue of $H$.
A matching upper bound to $S \gtrsim \Delta_{\text{gbl}}^{-1/2}$ is not known. However, in this section we build on the argument from \cite{AAG22} to show that $S\lesssim \Delta_{\text{loc}}^{-3/4 + o(1)}$.    Thus, the result in this section provides an improved dependence on the local gap for 1D frustration-free Hamiltonians, and brings us closer to the conjectured bound of $O(1/\bar{\Delta_{\text{gbl}}}^{1/2})$.

Given a local gap $\Delta := \Delta_{\text{loc}}$, we can approximate any ground state
projector $\Sigma_t$ on a subchain of length $t$ by a Chebyshev
polynomial \cite{arad2013area} of degree $d$ and with approximation
error $\|\Sigma_t - P_d(H_t)\|\leq e^{-\Theta(d\sqrt{\Delta/t})}$, with
$H_t$ the Hamiltonian restricted to the subchain. This approximation,
along with the Schmidt rank bounds from \cite{arad2013area} give us an
entanglement scaling of $\tilde{O}(1/\Delta)$. But we can further
improve the approximation using robust polynomials \cite{She12} and
the detectability lemma \cite{AALV09} (in particular, using the
coarse-grained version in \cite{PhysRevB.93.205142}): first
approximate $\Sigma_t$ as a product of projectors on smaller regions of
length $r<t$ as $\|\Sigma_t - \prod_i\Sigma^{(i)}_r\|_\infty \leq
e^{-\Theta(r\sqrt{\Delta})}$ (where $(i)$ an index of the smaller
projectors that ensures that the `odd' projectors show first and then
`even' projectors show) and then use the Chebyshev approximation along
with Sherstov's polynomial to approximate $\prod_i\Sigma^{(i)}_r$ by
a polynomial of degree $\frac{t}{\sqrt{r\Delta}}$ and with
approximation error $e^{-\Theta(\frac{t}{r})}$. Balancing the two
kinds of errors, we have $r^2= t/\sqrt{\Delta}$ and hence obtain an
approximation to $\Sigma_t$ of degree $t^{3/4}/\Delta^{3/8}$ and
approximation error $e^{-t^{3/4}\Delta^{1/8}}$. We need
$r>\frac{1}{\sqrt{\Delta}}$, meaning that the analysis above applies
when $t > \frac{1}{\sqrt{\Delta}}$.

We then recursively repeat this procedure to obtain better approximations. Suppose at the $j$-th level of recursion, we have found an approximation to $\Sigma_{m_j}$ of degree $\tilde{O}\left(\frac{m_j^{1-a_j}}{\Delta^{c_j}}\right)$ with error $e^{-m_j^{1-2f_j}\Delta^{\frac{1}{2}-f_j}}$ whenever $m_j \geq \frac{1}{\sqrt{\Delta}}$.  Here we have hidden several subpolynomial factors in the $\tilde{O}$ notation. Let us use it to obtain an improved approximation to $\Sigma_{m_{j+1}}$, with $m_{j+1}>m_j$. Write $\Sigma_{m_{j+1}}$ as a product of $\frac{m_{j+1}}{m_j}$ number of $\Sigma_{m_j}$, which has error $e^{-m_j\sqrt{\Delta}}$ due to the detectability lemma and the local gap assumption. Then apply Sherstov's robust method with degree $\frac{m_{j+1}}{m_j}$. The overall error is then
\begin{equation}
\exp\left(-m_j\sqrt{\Delta}\right)+ \exp\left(-\frac{m_{j+1}}{m_j} m_j^{1-2f_j}\Delta^{\frac{1}{2}-f_j}\right).
\end{equation}
Setting $m_j^{1+2f_j}\Delta^{f_j}=m_{j+1}$, the error becomes $e^{-m_j\sqrt{\Delta}}=\exp\left(-m^{1/(1+2f_j)}_{j+1}\Delta^{\frac{1}{2}-\frac{f_j}{1+2f_j}}\right)$. The total degree of approximation is 
\begin{equation}
\frac{m_{j+1}}{m_j}\cdot \frac{m_j^{1-a_j}}{\Delta^{c_j}}=\frac{m_j^{1+2f_j-a_j}}{\Delta^{c_j-f_j}}=\frac{m_{j+1}^{\frac{1+2f_j-a_j}{1+2f_j}}}{\Delta^{(c_j-f_j) + \frac{f_j(1+2f_j-a_j)}{1+2f_j}}}=\frac{m_{j+1}^{1-\frac{a_j}{1+2f_j}}}{\Delta^{c_j-a_jf_j/(1+2f_j)}}.
\end{equation}
The constraint $m_{j+1}>m_j$ is the same as $m_j^{2f_j}\geq \frac{1}{\Delta^{f_j}}$, which is satisfied since $m_j \geq \frac{1}{\sqrt{\Delta}}$ (this is the reason to choose the coefficient in the error as $m_j^{1-2f_j}$, which may seem arbitrary at first). Hence, we obtain a polynomial approximation to $\Sigma_{m_{j+1}}$ of degree $m_{j+1}^{1-a_{j+1}}/\Delta^{c_{j+1}}$, with error $\exp\left(-m_{j+1}^{1-2f_{j+1}}\Delta^{\frac{1}{2}-f_{j+1}}\right)$, where
\begin{equation}
    f_{j+1} = \frac{f_j}{1+2f_j},\quad a_{j+1}=\frac{a_j}{1+2f_j}, \quad c_{j+1}=c_j-\frac{a_jf_j}{1+2f_j}.
\end{equation}
Recall the initial condition $a_0=c_0=f_0=\frac{1}{2}$. 
The solution to above equations is
\begin{equation}
f_j = \frac{1}{2j+2}, \quad a_j=\frac{1}{2j+2}, \quad c_j = \frac{j+2}{4j+4}.
\end{equation}
In the limit of large $j$, we obtain the solutions $a_{*}=1, b_{*}=\frac{1}{2}, c_{*}=\frac{1}{4}, f_{*}=\frac{1}{2}$. Hence, $\Sigma_t$ can be approximated by a polynomial of degree $\tilde{O}\left(\frac{t}{\Delta^{1/4}}\right)$, with error $e^{-\Tilde{\Theta}(t\sqrt{\Delta})}$, where $\tilde{\Theta}$ hides sub-polynomial losses.

In order to obtain the desired entanglement bound, we now approximate the ground space with the product $\Sigma_t \Sigma_{<}\Sigma_{>}$, where $\Sigma_{<}$ and $\Sigma_{>}$ have an overlap of length $\frac{t}{4}$ with $\Sigma_t$, to the left and right of the cut. The approximation error is $e^{-\Theta(t\sqrt{\Delta})}$. Then, we approximate $\Sigma_{t}$ with the above polynomial and then raise to the power of $k$. This gives a $(D,\Delta)$-AGSP \cite{PhysRevB.85.195145} with overall shrinking $\Delta= e^{-\tilde{\Theta}(kt\sqrt{\Delta})}$ and the Schmidt rank $D=\tilde{O}(e^{k/\Delta^{1/4} + t})$. Setting $k=t\Delta^{1/4}$, we find that the shrinking is $e^{-\tilde{\Theta}(t^2\Delta^{3/4})}$ and Schmidt rank $\tilde{O}(e^{t})$. For the AGSP condition $D\Delta < \frac{1}{2}$, we set $t=\tilde{O}(1)/\Delta^{3/4}$, and get entropy bound of $\log D = \tilde{O}(1/\Delta^{3/4})$.  

\section{Spectral gap versus ground state correlation}
\label{append:CATlb}

For the Hamiltonian family constructed in~\cref{sec:vollaw}, the gap $\Theta(n^{-2})$ primarily followed from the gap of the propagation Hamiltonian. It is natural to ask if states like those in Eqn.~\ref{gsleft} require Hamiltonians with small gap. In this section, we show that the spectral gap must decay rapidly if the ground state has large correlation.

We note that an old result by Hastings and Koma~\cite{hastingskoma} establishes that, for Hamiltonians on bounded degree graphs with gap $\Delta$,  the ground state correlation function obeys
\begin{equation}
\mathrm{corr}(A, B) \leq C \norm{A} \norm{B} \exp(- d(A, B)/\xi)
\end{equation}
for $d(A,B)$ the graph distance and $\xi = O(\Delta^{-1})$.  Suppose the ground state is close to a cat state, and we can choose operators $A$ and $B$ such that $\mathrm{corr}(A, B) = O(1)$ (for instance $A = Z_i$ and $B = Z_j$).  This means that $\Delta \lesssim 1/d(A,B)$.  However, the diameter bound is poor: on an expander graph or tree, $d(A,B) \sim \log n$ where $n$ is the number of qubits.  Below, we will provide a stronger bound of $\Delta \leq \Theta(1/n)$ independent of the geometry using a different approach.

We illustrate the bound for cat state, but emphasize that it is more generally applicable.  For cat states, it is easy to see that $\ket{0^n} \pm \ket{1^n}$ have identical marginals on $\leq n-1$-qubit subsystems, and thus cannot be distinguished by any local Hamiltonian.  However, if the ground state approximates a cat state only on a subset of its qubits, then bounds on the gap scaling become interesting and nontrivial. 

\subsection{Upper bound on the gap}
Let $\ket{\psi}_{AB}$ be a quantum state on $n+m$ qubits, where $A$ consists of $n$ qubits and $B$ consists of the remaining $m$ qubits, such that  
\begin{equation}
    \|\psi_A - \ketbra{\CAT}_n\| \leq \frac{1}{10} \,.
\end{equation}
Here, $m$ may be arbitrarily large in comparison to $n$. Suppose $H=\sum_\alpha h_{\alpha}$ (with $\norm{h_{\alpha}} = O(1)$) is a $k$-local Hamiltonian of constant degree with spectral gap $\Delta$, and $\ket{\psi}_{AB}$ is a unique ground state of $H$. 
We truncate the $B$-only part of the Hamiltonian,
\begin{equation}
    H_B \coloneqq \sum_{\alpha: \ \supp(\alpha) \subseteq B} h_\alpha \,,
\end{equation}
and let $\widetilde{H}_B$ be a non-local operator supported on $B$ by retaining the $O(n+1/\Delta)$ lowest eigenvalues and eigenspaces of $H_B$. 
Define
\begin{equation}
    H' \coloneqq H-H_B + \widetilde{H}_B \,.
\end{equation}
As shown in \cite[Theorem 2.6]{AKL16}, the ground state $\ket{\psi'}_{AB}$ of $H'$ has fidelity at least $0.99$ with $\ket{\psi}_{AB}$, and hence $\|\psi'_A - \ketbra{\CAT}_n\| \leq 1/5$ and the spectral gap of $H'$ is $\Omega(\Delta)$. Also $\|H'\|=O(n)$.

Let $P_{d}$ be the shifted Chebyshev polynomial of degree $d$. We have 
$$\|\ketbra{\psi'}_{AB} - P_d(H')\|\leq 2 e^{-\Omega(1)\cdot d\sqrt{\Delta/n}}.$$ In particular, for $d=n/10k$, we have
$$\|\ketbra{\psi'}_{AB} - P_d(H')\|\leq 2 e^{-\Omega(1/k)\cdot \sqrt{n\Delta}}.$$
Now, we follow an argument from \cite{KAAV15} (made explicit in \cite[Lemma 2.10]{AM23}). Let $\Pi^{0}_A=\ketbra{0^{\otimes n}}$ and $\Pi^{1}_A=\ketbra{1^{\otimes n}}$. We know that $$\|\Pi^0_A\ketbra{\CAT}_n\Pi^1_A\| = \frac{1}{2}.$$ However, $\|\Pi^0_AP_d(H')\Pi^1_A\| = 0$ for $d=n/10k$, since the terms of $H'$ cannot flip all 0s to all 1s with degree at most $n/10$. This forces
$$2 e^{-\Omega(1/k)\cdot \sqrt{n\Delta}} \geq 1/2-1/5, \implies \Delta \leq O\!\left(\frac{k^2}{n}\right).$$

\begin{remark}
   The argument above can be applied to any ``CAT-like'' state, whenever one can find two well-separated projectors $\Pi^0_A$ and $\Pi^1_A$ that each have good overlap with the state.  For example, consider the state in Eqn.~\ref{gsleft}, call it $\ket{\phi}$.  Choose $\Pi_A^1 = \ketbra{\ast^{n-1}}{\ast^{n-1}}\otimes \ketbra{+}{+} \otimes \ketbra{0^{(\alpha - 1)n}}{0^{(\alpha - 1)n}}$ and $\Pi_A^0 = \ketbra{+}{+}\otimes \ketbra{0^{\alpha n-1}}{0^{\alpha n-1}}$.  Then, $\norm{\Pi_A^0 \ketbra{\phi}{\phi} \Pi_A^1} \simeq \frac{1}{n}$.  Since $\|\Pi^0_AP_d(H')\Pi^1_A\| = 0$ for $d=\Theta(n)$, this gives the bound $\Delta = O\left(\frac{\log^2 n}{n}\right)$.
\end{remark}
\begin{remark}\label{remark:FFgapcat}
   The CAT-state bound can be further improved when $H$ is 1D and frustration-free. It is known \cite{gosset2016correlation} that the correlation between two sites at distance $d$ apart in a ground state scales as $e^{- d \Theta(\sqrt{\Delta})}$. However, $\ket{\CAT}_n$ has correlation $1$ for any pair of qubits. This means $e^{- n \Theta(\sqrt{\Delta})} = \Omega(1)$ and hence $\Delta = O(1/n^2)$. 
\end{remark}

\subsection{Tightness of the bound}

For frustration-free Hamiltonians, we can achieve the bound in Remark~\ref{remark:FFgapcat} via a Feynman-Kitaev construction.  However, for the more general bound presented for the cat state and similar states that lack correlation decay, it is helpful to know about its tightness.

We provide a simple example demonstrating tightness of the bound.  Consider the ground state $\ket{\psi}$ of the XXZ Hamiltonian $H=\sum_i(X_iX_{i+1}+ Y_iY_{i+1}+\cos(\pi(1-\eta)) Z_iZ_{i+1})$ with spectral gap $\Theta(\eta^2/n)$, $\bra{\psi}(\sum_i (-1)^iX_i)\ket{\psi}=0$ and correlation $\bra{\psi}X_iX_j\ket{\psi}$ scaling as $\sim \frac{(-1)^{i-j}}{|i-j|^{\eta}}$ (see \cite{HF98} for an easily accessible and explicit formula, though this result was known far before then). This implies that 
\begin{equation}
V:=\sqrt{\bra{\psi}\left(\sum_i (-1)^iX_i\right)^2\ket{\psi}-\bra{\psi}\left(\sum_i (-1)^iX_i\right)\ket{\psi}^2} = \sqrt{n + \sum_{k=1}^{n-1} (n-k) k^{-\eta}}=\Omega(n^{1-\eta/2}).
\end{equation}
Let $\Pi_{\geq V/2}$ (similarly $\Pi_{\leq - V/2}$) be the projector onto the subspace of X-hamming weight $\geq V/2$ (similarly $\leq - V/2$). By symmetry, $\bra{\psi}\Pi_{\geq V/2}\ket{\psi} = \bra{\psi}\Pi_{\leq -V/2}\ket{\psi} \coloneq p$. Then,
\begin{equation}
V^2 \leq 2 p n^2 + (1-2p) \frac{V^2}{4}
\end{equation}
which implies $p \geq \frac{3V^2}{2(4n^2-V^2)}\geq \frac{3V^2}{8n^2} \sim n^{-\eta}$. Thus, the ground state has inverse-polynomial support on two sets which are Hamming distance $V/2$ apart, and the spectral gap is $O(1/n)$.  Applying the same arguments from the previous subsection, this shows that the well-separation technique is nearly tight. 

Though the bound is nearly tight, we do not know of a local Hamiltonian with a ground state that is close to the CAT state itself and a spectral gap $\Omega(1/n)$. One idea would be to develop a low-depth circuit to ``distill'' a cat state from the above CFT state and use this to construct a local parent Hamiltonian (possibly allowing for logarithmic-range interactions).  However, it is unclear to us how to implement the details.

\section*{Acknowledgments}
All ideas and proofs were developed by the authors without the use of
an LLM.  
We acknowledge use of ChatGPT 5.6 only for additional
verification of our results and for suggesting a cleaner notation
in~\cref{appendix:unitary}
and~\cref{sec:vollaw}.  
We thank David P\'erez-Garc\'ia for useful
conversations about MPOs and Jonas Helsen for discussions on gapped generator sets of the Clifford group. 
AA acknowledges support through the NSF
Award Nos. 2238836, 2430375 and QCIS-FF: Quantum Computing \&
Information Science Faculty Fellow at Harvard University (NSF
2013303).  AWH acknowledges support from the Simons Foundation
(MP-SIP-00001553, AWH) and the Department of Energy through the grant
DE-SC0020360 ``Complex quantum systems and the quantum universe''.  SB
was supported by the Walter Burke Institute for Theoretical Physics at
Caltech and by the Institute for Quantum Information and Matter, an
NSF Physics Frontiers Center (NSF Grant PHY-2317110).  XT is supported
by the U.S. Department of Energy, Office of Science, National Quantum
Information Science Research Centers, Co-design Center for Quantum
Advantage (C2QA) under contract number DE-SC0012704.

\newcommand{\etalchar}[1]{$^{#1}$}

\appendix


\begin{thebibliography}{ICPGSV17}

\bibitem[AAG22]{AAG22}
Anurag Anshu, Itai Arad, and David Gosset.
\newblock An area law for 2d frustration-free spin systems.
\newblock In {\em Proceedings of the 54th Annual ACM SIGACT Symposium on Theory
  of Computing}, STOC 2022, page 129–137, New York, NY, USA, 2022.
  Association for Computing Machinery,
  \href{http://arxiv.org/abs/2103.02492}{{\ttfamily arXiv:2103.02492}}.

\bibitem[AALV09]{AALV09}
Dorit Aharonov, Itai Arad, Zeph Landau, and Umesh Vazirani.
\newblock The detectability lemma and quantum gap amplification.
\newblock In {\em Proceedings of the Forty-First Annual ACM Symposium on Theory
  of Computing}, STOC '09, page 417–426, New York, NY, USA, 2009. Association
  for Computing Machinery,  \href{http://arxiv.org/abs/0811.3412}{{\ttfamily
  arXiv:0811.3412}}.

\bibitem[AAV16]{PhysRevB.93.205142}
Anurag Anshu, Itai Arad, and Thomas Vidick.
\newblock Simple proof of the detectability lemma and spectral gap
  amplification.
\newblock {\em Phys. Rev. B}, 93:205142, May 2016,
  \href{http://arxiv.org/abs/1602.01210}{{\ttfamily arXiv:1602.01210}}.

\bibitem[AG04]{aaronson2004improved}
Scott Aaronson and Daniel Gottesman.
\newblock Improved simulation of stabilizer circuits.
\newblock {\em Physical Review A—Atomic, Molecular, and Optical Physics},
  70(5):052328, 2004,  \href{http://arxiv.org/abs/quant-ph/0406196}{{\ttfamily
  arXiv:quant-ph/0406196}}.

\bibitem[AHL{\etalchar{+}}14]{aharonov2014local}
Dorit Aharonov, Aram~W Harrow, Zeph Landau, Daniel Nagaj, Mario Szegedy, and
  Umesh Vazirani.
\newblock Local tests of global entanglement and a counterexample to the
  generalized area law.
\newblock In {\em 2014 IEEE 55th Annual Symposium on Foundations of Computer
  Science}, pages 246--255. IEEE, 2014,
  \href{http://arxiv.org/abs/1410.0951}{{\ttfamily arXiv:1410.0951}}.

\bibitem[AK62]{Arnold62}
V.~I. Arnold and A.~L. Krylov.
\newblock Uniform distribution of points on a sphere and some ergodic
  properties of linear ordinary differential equations in the complex plane.
\newblock {\em Soviet Math. Dokl.}, 4:1--5, 1962.

\bibitem[AKL16]{AKL16}
Itai Arad, Tomotaka Kuwahara, and Zeph Landau.
\newblock Connecting global and local energy distributions in quantum spin
  models on a lattice.
\newblock {\em Journal of Statistical Mechanics: Theory and Experiment},
  2016(3):033301, 2016,  \href{http://arxiv.org/abs/1406.3898}{{\ttfamily
  arXiv:1406.3898}}.

\bibitem[AKLV13]{arad2013area}
Itai Arad, Alexei Kitaev, Zeph Landau, and Umesh Vazirani.
\newblock An area law and sub-exponential algorithm for 1d systems, 2013,
  \href{http://arxiv.org/abs/1301.1162}{{\ttfamily arXiv:1301.1162}}.

\bibitem[ALV12]{PhysRevB.85.195145}
Itai Arad, Zeph Landau, and Umesh Vazirani.
\newblock Improved one-dimensional area law for frustration-free systems.
\newblock {\em Phys. Rev. B}, 85:195145, May 2012,
  \href{http://arxiv.org/abs/1111.2970}{{\ttfamily arXiv:1111.2970}}.

\bibitem[AM23]{AM23}
Anurag Anshu and Tony Metger.
\newblock Concentration bounds for quantum states and limitations on the {QAOA}
  from polynomial approximations.
\newblock {\em {Quantum}}, 7:999, May 2023,
  \href{http://arxiv.org/abs/2209.02715}{{\ttfamily arXiv:2209.02715}}.

\bibitem[Ans20]{anshu2020improved}
Anurag Anshu.
\newblock Improved local spectral gap thresholds for lattices of finite size.
\newblock {\em Physical Review B}, 101(16):165104, 2020,
  \href{http://arxiv.org/abs/1909.01516}{{\ttfamily arXiv:1909.01516}}.

\bibitem[AS04]{ambainis2004small}
Andris Ambainis and Adam Smith.
\newblock Small pseudo-random families of matrices: Derandomizing approximate
  quantum encryption.
\newblock In {\em International Workshop on Randomization and Approximation
  Techniques in Computer Science}, pages 249--260. Springer, 2004,
  \href{http://arxiv.org/abs/quant-ph/0404075}{{\ttfamily
  arXiv:quant-ph/0404075}}.

\bibitem[BASTS10]{ben2010quantum}
Avraham Ben-Aroya, Oded Schwartz, and Amnon Ta-Shma.
\newblock Quantum expanders: Motivation and construction.
\newblock {\em Theory of Computing}, 6(1):47--79, 2010,
  \href{http://arxiv.org/abs/0709.0911}{{\ttfamily arXiv:0709.0911}}.

\bibitem[BATS07]{ben2007quantum}
Avraham Ben-Aroya and Amnon Ta-Shma.
\newblock Quantum expanders and the quantum entropy difference problem, 2007,
  \href{http://arxiv.org/abs/quant-ph/0702129}{{\ttfamily
  arXiv:quant-ph/0702129}}.

\bibitem[BCG{\etalchar{+}}02]{barnum2002authentication}
Howard Barnum, Claude Cr{\'e}peau, Daniel Gottesman, Adam Smith, and Alain
  Tapp.
\newblock Authentication of quantum messages.
\newblock In {\em The 43rd Annual IEEE Symposium on Foundations of Computer
  Science, 2002. Proceedings.}, pages 449--458. IEEE, 2002,
  \href{http://arxiv.org/abs/quant-ph/0205128}{{\ttfamily
  arXiv:quant-ph/0205128}}.

\bibitem[BDSW96]{bennett1996mixed}
Charles~H Bennett, David~P DiVincenzo, John~A Smolin, and William~K Wootters.
\newblock Mixed-state entanglement and quantum error correction.
\newblock {\em Physical Review A}, 54(5):3824, 1996,
  \href{http://arxiv.org/abs/quant-ph/9604024}{{\ttfamily
  arXiv:quant-ph/9604024}}.

\bibitem[BG08]{bourgain2008spectral}
Jean Bourgain and Alex Gamburd.
\newblock On the spectral gap for finitely-generated subgroups of {$SU(2)$}.
\newblock {\em Inventiones mathematicae}, 171(1):83--121, 2008.

\bibitem[BG12]{bourgain2012spectral}
Jean Bourgain and Alex Gamburd.
\newblock A spectral gap theorem in {SU}(d).
\newblock {\em Journal of the European Mathematical Society (EMS Publishing)},
  14(5):1455, 2012,  \href{http://arxiv.org/abs/1108.6264}{{\ttfamily
  arXiv:1108.6264}}.

\bibitem[BH26]{baer2026random}
Tim Baer and Jeongwan Haah.
\newblock Random unitary circuits with constant spectral gap, 2026,
  \href{http://arxiv.org/abs/2607.20919}{{\ttfamily arXiv:2607.20919}}.

\bibitem[BHH16]{brandao2016local}
Fernando~GSL Brandao, Aram~W Harrow, and Micha{\l} Horodecki.
\newblock Local random quantum circuits are approximate polynomial-designs.
\newblock {\em Communications in Mathematical Physics}, 346(2):397--434, 2016,
  \href{http://arxiv.org/abs/1208.0692}{{\ttfamily arXiv:1208.0692}}.

\bibitem[Bou17]{bourgain2017random}
Jean Bourgain.
\newblock On random walks in large compact {L}ie groups.
\newblock In {\em Geometric Aspects of Functional Analysis: Israel Seminar
  (GAFA) 2014--2016}, pages 55--63. Springer, 2017,
  \href{http://arxiv.org/abs/1501.01597}{{\ttfamily arXiv:1501.01597}}.

\bibitem[CHH{\etalchar{+}}24]{CHHLMT24}
Chi-Fang Chen, Jeongwan Haah, Jonas Haferkamp, Yunchao Liu, Tony Metger, and
  Xinyu Tan.
\newblock Incompressibility and spectral gaps of random circuits, 2024,
  \href{http://arxiv.org/abs/2406.07478}{{\ttfamily arXiv:2406.07478}}.

\bibitem[CN18]{caha2018pair}
Libor Caha and Daniel Nagaj.
\newblock The pair-flip model: a very entangled translationally invariant spin
  chain, 2018,  \href{http://arxiv.org/abs/1805.07168}{{\ttfamily
  arXiv:1805.07168}}.

\bibitem[FVZS{\etalchar{+}}18]{Fishman_2018}
M.~T. Fishman, L.~Vanderstraeten, V.~Zauner-Stauber, J.~Haegeman, and
  F.~Verstraete.
\newblock Faster methods for contracting infinite two-dimensional tensor
  networks.
\newblock {\em Physical Review B}, 98(23), December 2018,
  \href{http://arxiv.org/abs/1711.05881}{{\ttfamily arXiv:1711.05881}}.

\bibitem[GE08]{gross2007quantum}
David Gross and Jens Eisert.
\newblock Quantum {M}argulis expanders.
\newblock {\em Quant. Inf. Comp}, 8:722, 2008,
  \href{http://arxiv.org/abs/0710.0651}{{\ttfamily arXiv:0710.0651}}.

\bibitem[GH10]{gottesman2010entanglement}
Daniel Gottesman and Matthew~B Hastings.
\newblock Entanglement versus gap for one-dimensional spin systems.
\newblock {\em New journal of physics}, 12(2):025002, 2010,
  \href{http://arxiv.org/abs/0901.1108}{{\ttfamily arXiv:0901.1108}}.

\bibitem[GH16]{gosset2016correlation}
David Gosset and Yichen Huang.
\newblock Correlation length versus gap in frustration-free systems.
\newblock {\em Physical review letters}, 116(9):097202, 2016,
  \href{http://arxiv.org/abs/1509.06360}{{\ttfamily arXiv:1509.06360}}.

\bibitem[GM16]{gosset2016local}
David Gosset and Evgeny Mozgunov.
\newblock Local gap threshold for frustration-free spin systems.
\newblock {\em Journal of Mathematical Physics}, 57(9), 2016,
  \href{http://arxiv.org/abs/1512.00088}{{\ttfamily arXiv:1512.00088}}.

\bibitem[Har07]{harrow2007quantum}
Aram~W Harrow.
\newblock Quantum expanders from any classical {C}ayley graph expander.
\newblock {\em Q. Inf. Comp., vol. 8, no. 8/9, pp. 715-721, 2008.}, 2007,
  \href{http://arxiv.org/abs/0709.1142}{{\ttfamily arXiv:0709.1142}}.

\bibitem[Has04]{Hastings04}
M.~B. Hastings.
\newblock {Lieb-Schultz-Mattis} in higher dimensions.
\newblock {\em Phys. Rev. B}, 69:104431, Mar 2004,
  \href{http://arxiv.org/abs/cond-mat/0305505}{{\ttfamily
  arXiv:cond-mat/0305505}}.

\bibitem[Has07a]{PhysRevA.76.032315}
M.~B. Hastings.
\newblock Random unitaries give quantum expanders.
\newblock {\em Phys. Rev. A}, 76:032315, Sep 2007,
  \href{http://arxiv.org/abs/0706.0556}{{\ttfamily arXiv:0706.0556}}.

\bibitem[Has07b]{hastings2007entropy}
Matthew~B Hastings.
\newblock Entropy and entanglement in quantum ground states.
\newblock {\em Physical Review B}, 76(3):035114, 2007,
  \href{http://arxiv.org/abs/cond-mat/0701055}{{\ttfamily
  arXiv:cond-mat/0701055}}.

\bibitem[{Has}14]{Hastings14}
M.~B. {Hastings}.
\newblock {Notes on Some Questions in Mathematical Physics and Quantum
  Information}, 2014,  \href{http://arxiv.org/abs/1404.4327}{{\ttfamily
  arXiv:1404.4327}}.

\bibitem[HF98]{HF98}
T.~Hikihara and A.~Furusaki.
\newblock Correlation amplitude for the $s=\frac{1}{2}$ $\mathrm{XXZ}$ spin
  chain in the critical region: Numerical renormalization-group study of an
  open chain.
\newblock {\em Phys. Rev. B}, 58:R583(R)--R586(R), Jul 1998.

\bibitem[HH09]{HH08}
Matthew~B Hastings and Aram~W Harrow.
\newblock Classical and quantum tensor product expanders.
\newblock {\em Q. Inf. Comp.}, 9(3\&4):336--360, 2009,
  \href{http://arxiv.org/abs/0804.0011}{{\ttfamily arXiv:0804.0011}}.

\bibitem[HK06]{hastingskoma}
Matthew~B. Hastings and Tohru Koma.
\newblock Spectral gap and exponential decay of correlations.
\newblock {\em Communications in Mathematical Physics}, 265(3):781--804, 2006,
  \href{http://arxiv.org/abs/math-ph/0507008}{{\ttfamily
  arXiv:math-ph/0507008}}.

\bibitem[HL08]{harrow2008communication}
Aram~W Harrow and Debbie~W Leung.
\newblock A communication-efficient nonlocal measurement with application to
  communication complexity and bipartite gate capacities.
\newblock {\em IEEE Trans. Info. Theory, Volume 57, Issue 8, pages 5504 - 5508
  (Aug 2011)}, 2008,  \href{http://arxiv.org/abs/0803.3066}{{\ttfamily
  arXiv:0803.3066}}.

\bibitem[HLT25]{HLT25}
Jeongwan Haah, Yunchao Liu, and Xinyu Tan.
\newblock Efficient approximate unitary designs from random {Pauli} rotations.
\newblock {\em Communications in Mathematical Physics}, 406(12), Oct 2025,
  \href{http://arxiv.org/abs/2402.05239}{{\ttfamily arXiv:2402.05239}}.

\bibitem[HRT25]{HRT25}
Zhiyang He, Luke Robitaille, and Xinyu Tan.
\newblock Characterization of permutation gates in the third level of the
  clifford hierarchy, 2025,  \href{http://arxiv.org/abs/2510.04993}{{\ttfamily
  arXiv:2510.04993}}.

\bibitem[ICPGSV17]{CPSV17}
J~Ignacio~Cirac, David Perez-Garcia, Norbert Schuch, and Frank Verstraete.
\newblock Matrix product unitaries: structure, symmetries, and topological
  invariants.
\newblock {\em Journal of Statistical Mechanics: Theory and Experiment},
  2017(8):083105, August 2017,
  \href{http://arxiv.org/abs/1703.09188}{{\ttfamily arXiv:1703.09188}}.

\bibitem[IL26]{ippoliti2025infinite}
Matteo Ippoliti and David~M Long.
\newblock Infinite temperature at zero energy.
\newblock {\em Phys. Rev. X}, 16:031030, 2026,
  \href{http://arxiv.org/abs/2509.04410}{{\ttfamily arXiv:2509.04410}}.

\bibitem[Ira10]{irani2010ground}
Sandy Irani.
\newblock Ground state entanglement in one-dimensional translationally
  invariant quantum systems.
\newblock {\em Journal of Mathematical Physics}, 51(2):022101, 2010,
  \href{http://arxiv.org/abs/0901.1107}{{\ttfamily arXiv:0901.1107}}.

\bibitem[JNV{\etalchar{+}}21]{ji2021mip}
Zhengfeng Ji, Anand Natarajan, Thomas Vidick, John Wright, and Henry Yuen.
\newblock {MIP*=RE}.
\newblock {\em Communications of the ACM}, 64(11):131--138, 2021,
  \href{http://arxiv.org/abs/2001.04383}{{\ttfamily arXiv:2001.04383}}.

\bibitem[KAA21]{kuwahara2021improved}
Tomotaka Kuwahara, {\'A}lvaro~M Alhambra, and Anurag Anshu.
\newblock Improved thermal area law and quasilinear time algorithm for quantum
  gibbs states.
\newblock {\em Physical Review X}, 11(1):011047, 2021,
  \href{http://arxiv.org/abs/2007.11174}{{\ttfamily arXiv:2007.11174}}.

\bibitem[KAAV17]{KAAV15}
Tomotaka Kuwahara, Itai Arad, Luigi Amico, and Vlatko Vedral.
\newblock Local reversibility and entanglement structure of many-body ground
  states.
\newblock {\em Quantum Science and Technology}, 2(1):015005, 2017,
  \href{http://arxiv.org/abs/1502.05330}{{\ttfamily arXiv:1502.05330}}.

\bibitem[Kas07a]{kassabov2007symmetric}
Martin Kassabov.
\newblock Symmetric groups and expander graphs.
\newblock {\em Inventiones mathematicae}, 170(2):327--354, 2007,
  \href{http://arxiv.org/abs/math/0505624}{{\ttfamily arXiv:math/0505624}}.

\bibitem[Kas07b]{Kas07a}
Martin Kassabov.
\newblock Universal lattices and unbounded rank expanders.
\newblock {\em Inventiones mathematicae}, 170(2):297--326, 2007.

\bibitem[KLN06]{kassabov2006finite}
Martin Kassabov, Alexander Lubotzky, and Nikolay Nikolov.
\newblock Finite simple groups as expanders.
\newblock {\em Proceedings of the National Academy of Sciences},
  103(16):6116--6119, 2006,
  \href{http://arxiv.org/abs/math/0510562}{{\ttfamily arXiv:math/0510562}}.

\bibitem[KT13]{KastoryanoT13}
Michael~J. Kastoryano and Kristan Temme.
\newblock Quantum logarithmic {S}obolev inequalities and rapid mixing.
\newblock {\em Journal of Mathematical Physics}, 54(5), 2013,
  \href{http://arxiv.org/abs/1207.3261}{{\ttfamily arXiv:1207.3261}}.

\bibitem[LCB14]{Lubasch_2014}
Michael Lubasch, J~Ignacio Cirac, and Mari-Carmen Bañuls.
\newblock Unifying projected entangled pair state contractions.
\newblock {\em New Journal of Physics}, 16(3):033014, March 2014,
  \href{http://arxiv.org/abs/1311.6696}{{\ttfamily arXiv:1311.6696}}.

\bibitem[LH26]{liu2026almost}
Guoding Liu and Jonas Helsen.
\newblock (almost) quadruply optimal unitary designs in {1D}, 2026,
  \href{http://arxiv.org/abs/2608.18650}{{\ttfamily arXiv:2608.18650}}.

\bibitem[LL25]{lemm2024critical}
Marius Lemm and Angelo Lucia.
\newblock On the critical finite-size gap scaling for frustration-free
  {H}amiltonians.
\newblock {\em Reviews in Mathematical Physics, Vol. 37, No}, 7:2550015, 2025,
  \href{http://arxiv.org/abs/2409.09685}{{\ttfamily arXiv:2409.09685}}.

\bibitem[Lub11]{lubotzky2011finite}
Alexander Lubotzky.
\newblock Finite simple groups of {L}ie type as expanders.
\newblock {\em Journal of the European Mathematical Society}, 13(5):1331--1341,
  2011,  \href{http://arxiv.org/abs/0904.3411}{{\ttfamily arXiv:0904.3411}}.

\bibitem[MPSY24]{metger2024simple}
Tony Metger, Alexander Poremba, Makrand Sinha, and Henry Yuen.
\newblock Simple constructions of linear-depth t-designs and pseudorandom
  unitaries.
\newblock In {\em 2024 IEEE 65th Annual Symposium on Foundations of Computer
  Science (FOCS)}, pages 485--492. IEEE, 2024,
  \href{http://arxiv.org/abs/2404.12647}{{\ttfamily arXiv:2404.12647}}.

\bibitem[MS16]{movassagh2016supercritical}
Ramis Movassagh and Peter~W Shor.
\newblock Supercritical entanglement in local systems: Counterexample to the
  area law for quantum matter.
\newblock {\em Proceedings of the National Academy of Sciences},
  113(47):13278--13282, 2016,  \href{http://arxiv.org/abs/1408.1657}{{\ttfamily
  arXiv:1408.1657}}.

\bibitem[Nik05]{nikolov2005product}
Nikolay Nikolov.
\newblock A product decomposition for the classical quasisimple groups, 2005,
  \href{http://arxiv.org/abs/math/0510173}{{\ttfamily arXiv:math/0510173}}.

\bibitem[OSP23]{o2023explicit}
Ryan O’Donnell, Rocco~A Servedio, and Pedro Paredes.
\newblock Explicit orthogonal and unitary designs.
\newblock In {\em 2023 IEEE 64th Annual Symposium on Foundations of Computer
  Science (FOCS)}, pages 1240--1260. IEEE, 2023,
  \href{http://arxiv.org/abs/2310.13597}{{\ttfamily arXiv:2310.13597}}.

\bibitem[Sar15]{sarnak2015letter}
Peter Sarnak.
\newblock {Letter to Scott Aaronson and Andy Pollington on the Solovay-Kitaev
  Theorem and Golden Gates}, 2015.

\bibitem[She12]{She12}
Alexander~A. Sherstov.
\newblock Making polynomials robust to noise.
\newblock In {\em Proceedings of the Forty-Fourth Annual ACM Symposium on
  Theory of Computing}, STOC '12, page 747–758, New York, NY, USA, 2012.
  Association for Computing Machinery.

\bibitem[SML{\etalchar{+}}25]{schuster2025strong}
Thomas Schuster, Fermi Ma, Alex Lombardi, Fernando Brandao, and Hsin-Yuan
  Huang.
\newblock Strong random unitaries and fast scrambling, 2025,
  \href{http://arxiv.org/abs/2509.26310}{{\ttfamily arXiv:2509.26310}}.

\bibitem[TKR{\etalchar{+}}10]{chi2-divergence}
K.~Temme, M.~J. Kastoryano, M.~B. Ruskai, M.~M. Wolf, and F.~Verstraete.
\newblock The {$\chi^2$}-divergence and mixing times of quantum {M}arkov
  processes.
\newblock {\em Journal of Mathematical Physics}, 51(12):122201, 2010,
  \href{http://arxiv.org/abs/1005.2358}{{\ttfamily arXiv:1005.2358}}.

\bibitem[Wie24]{wierichs2024kak}
David Wierichs.
\newblock The {KAK} decomposition.
\newblock PennyLane Demos, November 2024.

\bibitem[WVHC08]{wolf2008area}
Michael~M Wolf, Frank Verstraete, Matthew~B Hastings, and J~Ignacio Cirac.
\newblock Area laws in quantum systems: mutual information and correlations.
\newblock {\em Physical review letters}, 100(7):070502, 2008,
  \href{http://arxiv.org/abs/0704.3906}{{\ttfamily arXiv:0704.3906}}.

\bibitem[ZAK17]{klich}
Zhao Zhang, Amr Ahmadain, and Israel Klich.
\newblock Novel quantum phase transition from bounded to extensive
  entanglement.
\newblock {\em Proceedings of the National Academy of Sciences}, 114(20), 2017,
   \href{http://arxiv.org/abs/1606.07795}{{\ttfamily arXiv:1606.07795}}.

\end{thebibliography}
\end{document}